\documentclass[sigconf]{aamas} 

\pdfoutput=1

\usepackage{balance}
\usepackage{hyperref}
\usepackage{cleveref}
\usepackage{ dsfont }
\usepackage{amsmath}
\usepackage{stmaryrd}
\usepackage{scalefnt}
\usepackage{tikz}

\usepackage{mathtools}

\usepackage{algorithm}
\usepackage{listings}

\usepackage{relsize}
\usepackage{enumitem}
\usepackage{amsthm}
\usepackage{thmtools}

\newtheoremstyle{mystyle}%
{}{}%
{\itshape}%
{}%
{\scshape}
{.}%
{.5em}%
{\indent\thmname{#1}\thmnumber{ #2}\thmnote{ (#3)}}

\declaretheorem[style=mystyle]{lemma}
\declaretheorem[style=mystyle]{theorem}
\declaretheorem[style=mystyle]{example}
\declaretheorem[style=mystyle]{remark}
\declaretheorem[style=mystyle]{proposition}
\declaretheorem[style=mystyle]{definition}
\declaretheorem[style=mystyle]{corollary}

\newcommand{\EXPTIME}{\texttt{EXPTIME}}

\newcommand{\EXPSPACE}{\texttt{EXPSPACE}}

\newcommand{\ATLS}{ATL$^*$}
\newcommand{\CTLS}{CTL$^*$}

\newcommand{\SL}{SL}

\newcommand{\LTL}{LTL}
\newcommand{\HyperSL}{HyperSL}
\newcommand{\HyperSLP}{HyperSL[SPE]}

\newcommand{\HyperATLS}{HyperATL$^*$}
\newcommand{\HyperATLSS}{HyperATL$^*_S$}

\newcommand{\HyperLTL}{HyperLTL}
\newcommand{\HyperCTLS}{HyperCTL$^*$}
\newcommand{\SLI}{SL$_{\mathit{ii}}$}

\newcommand{\tool}{\texttt{HyMASMC}}
\newcommand{\mcmas}{\texttt{MCMAS}}
\newcommand{\mcmassl}{\texttt{MCMAS-SL[1G]}}

\newcommand{\bE}{{\boldsymbol{e}}}
\newcommand{\bO}{{\boldsymbol{o}}}
\newcommand{\bF}{{\vec{f}}}
\newcommand{\bA}{{\vec{a}}}

\newcommand{\bS}{{\boldsymbol{s}}}

\newcommand{\bT}{{\vec{t}}}
\newcommand{\bL}{{\boldsymbol{l}}}

\newcommand{\bQ}{{\boldsymbol{q}}}

\newcommand{\bX}{{\vec{x}}}

\newcommand{\stratVars}{\mathcal{X}}

\DeclareMathOperator{\ltlN}{\normalfont\textsf{X}}
\DeclareMathOperator{\ltlG}{\normalfont\textsf{G}}
\DeclareMathOperator{\ltlF}{\normalfont\textsf{F}}

\DeclareMathOperator{\ltlU}{\normalfont\textsf{U}}
\DeclareMathOperator{\ltlW}{\normalfont\textsf{W}}

\newcommand{\ap}{\mathit{AP}}
\newcommand{\pathVars}{\mathcal{V}}

\newcommand{\ldot}{\mathpunct{.}}

\newcommand{\quant}{\mathds{Q}}

\newcommand{\nat}{\mathbb{N}}
\newcommand{\bool}{\mathbb{B}}

\newcommand{\calG}{\mathcal{G}}
\newcommand{\calA}{\mathcal{A}}
\newcommand{\calB}{\mathcal{B}}
\newcommand{\calL}{\mathcal{L}}

\newcommand{\calO}{\mathcal{O}}

\newcommand{\calP}{\mathcal{P}}

\newcommand{\obs}{\mathit{Obs}}

\newcommand{\agents}{\mathit{{Agts}}}
\newcommand{\moves}{\mathbb{A}}

\newcommand{\play}[0]{\mathit{Play}}
\newcommand{\strats}[1]{\mathit{Str}(#1)}

\newcommand{\posController}[1]{\mathit{PosStr}(#1)}
\newcommand{\skoController}[1]{\mathit{SkoStr}(#1)}

\makeatletter
\DeclareFontFamily{OMX}{MnSymbolE}{}
\DeclareSymbolFont{MnLargeSymbols}{OMX}{MnSymbolE}{m}{n}
\SetSymbolFont{MnLargeSymbols}{bold}{OMX}{MnSymbolE}{b}{n}
\DeclareFontShape{OMX}{MnSymbolE}{m}{n}{
	<-6>  MnSymbolE5
	<6-7>  MnSymbolE6
	<7-8>  MnSymbolE7
	<8-9>  MnSymbolE8
	<9-10> MnSymbolE9
	<10-12> MnSymbolE10
	<12->   MnSymbolE12
}{}
\DeclareFontShape{OMX}{MnSymbolE}{b}{n}{
	<-6>  MnSymbolE-Bold5
	<6-7>  MnSymbolE-Bold6
	<7-8>  MnSymbolE-Bold7
	<8-9>  MnSymbolE-Bold8
	<9-10> MnSymbolE-Bold9
	<10-12> MnSymbolE-Bold10
	<12->   MnSymbolE-Bold12
}{}

\let\llangle\@undefined
\let\rrangle\@undefined
\DeclareMathDelimiter{\llangle}{\mathopen}%
{MnLargeSymbols}{'164}{MnLargeSymbols}{'164}
\DeclareMathDelimiter{\rrangle}{\mathclose}%
{MnLargeSymbols}{'171}{MnLargeSymbols}{'171}
\makeatother

\makeatletter
\DeclareRobustCommand\bigop[1]{%
	\mathop{\vphantom{\sum}\mathpalette\bigop@{#1}}\slimits@
}
\newcommand{\bigop@}[2]{%
	\vcenter{%
		\sbox\z@{$#1\sum$}%
		\hbox{\resizebox{\ifx#1\displaystyle.9\fi\dimexpr\ht\z@+\dp\z@}{!}{$\m@th#2$}}%
	}%
}
\makeatother

\makeatletter
\newcommand{\superimpose}[2]{{%
		\ooalign{%
			\hfil$\m@th#1\@firstoftwo#2$\hfil\cr
			\hfil$\m@th#1\@secondoftwo#2$\hfil\cr
		}%
}}
\makeatother

\newcommand{\landlor}{\mathrlap{\vee}\wedge}

\newcommand{\biglandlor}{\DOTSB\bigop{\landlor}}

\newcommand\ScaleExists[1]{\vcenter{\hbox{\scalefont{#1}$\exists$}}}
\newcommand\ScaleForall[1]{\vcenter{\hbox{\scalefont{#1}$\forall$}}}

\DeclareMathOperator*\bigexists{%
	\vphantom\sum
	\mathchoice{\ScaleExists{1.7}}{\ScaleExists{1.4}}{\ScaleExists{1}}{\ScaleExists{0.75}}}

\DeclareMathOperator*\bigforall{%
	\vphantom\sum
	\mathchoice{\ScaleForall{1.7}}{\ScaleForall{1.4}}{\ScaleForall{1}}{\ScaleForall{0.75}}}

\definecolor{mydarkblue}{HTML}{254E94}
\definecolor{mydarkblued}{HTML}{142B52}

\definecolor{mydarkred}{HTML}{7D3511}
\definecolor{mydarkredd}{HTML}{59260C}

\definecolor{mydarkviolet}{HTML}{4A3078}

\definecolor{mydarkgreen}{HTML}{104A19}
\definecolor{mydarkgreend}{HTML}{104A19}

\definecolor{mydarkyellow}{HTML}{87660C}

\definecolor{mydarkb}{HTML}{027d6d}

\newcommand{\agent}[1]{{\color{mydarkviolet}#1}}

\newcommand{\psiSLI}{{\psi}}
\newcommand{\varphiSLI}{{\varphi}}

\newcommand{\psiH}{{\psi}}
\newcommand{\varphiH}{{\varphi}}

\newcommand{\psiSL}{{\psi}}
\newcommand{\varphiSL}{{\varphi}}

\newcommand{\psiATL}{{\psi}}
\newcommand{\varphiATL}{{\varphi}}

\newcommand{\slToHyper}[1]{{\llparenthesis} #1 {\rrparenthesis}}

\newcommand{\slIToHyper}[1]{{\llparenthesis} #1 {\rrparenthesis}}

\newcommand{\AtlToHyper}[1]{{\llparenthesis} #1 {\rrparenthesis}}

\definecolor{dkcyan}{rgb}{0.1, 0.3, 0.3}
\definecolor{dkgreen}{rgb}{0,0.3,0}
\definecolor{olive}{rgb}{0.5, 0.5, 0.0}
\definecolor{dkblue}{rgb}{0,0.1,0.5}

\definecolor{col:ln}{rgb}  {0.1, 0.1, 0.7}
\definecolor{col:str}{rgb} {0.8, 0.0, 0.0}
\definecolor{col:db}{rgb}  {0.9, 0.5, 0.0}
\definecolor{col:ours}{rgb}{0.0, 0.7, 0.0}

\definecolor{lightgreen}{RGB}{170, 255, 220}
\definecolor{darkbrown}{RGB}{121,37,0}

\colorlet{listing-comment}{gray}
\colorlet{operator-color}{darkbrown}
\colorlet{comment-color}{black!50}

\lstdefinelanguage{custom-lang}{
	keywords={let, rec, in, match, with, when, if, then, else, elif, for, to, do, return, from, where, and, def},
	keywordstyle=[1]\bfseries,
	morekeywords=[2]{modelCheck},
	keywordstyle=[2]\color{dkgreen},
	morekeywords=[3]{simulate,nestedStateFormulas,LTLtoAPA,constructParityGame,solveParityGame,modelCheckO, toDPA},
	keywordstyle=[3]\color{dkcyan},
	comment=[l][\color{comment-color}]{//},
	literate=%
	{=}{$\boldsymbol{=}$}1
	{<-}{{{\color{operator-color}$\leftarrow$}}}1
	{|}{{{\color{dkblue}|}}}1
	{:}{$\boldsymbol{:}$}1
	{@}{ }1
}

\lstdefinestyle{default}{
	escapeinside={(*}{*)},
	basicstyle=\ttfamily\fontsize{8}{9.6}\selectfont,
	columns=fullflexible,
	commentstyle=\sffamily\color{black!50!white},
	framexleftmargin=1em,
	framexrightmargin=1ex,
	keepspaces=true,
	keywordstyle=\color{dkblue},
	mathescape,
	numbers=left,
	numberblanklines=false,
	numbersep=0.5em,
	numberstyle=\relscale{0.75}\color{gray}\ttfamily,
	showstringspaces=true,
	stepnumber=1,
	xleftmargin=1.2em,
}

\lstnewenvironment{code}[1][]
{\small
	\lstset{
		style=default, 
		language=custom-lang,
		#1
	}
}
{}

\newcommand\xqed[1]{%
	\leavevmode\unskip\penalty9999 \hbox{}\nobreak\hfill
	\quad\hbox{#1}}
\newcommand\demo{\xqed{$\triangle$}}

\newif\iffullversion
\fullversiontrue

\newcommand{\ifFull}[2]{\iffullversion#1\else#2\fi}

\setcopyright{ifaamas}
\acmConference[AAMAS '24]{Proc.\@ of the 23rd International Conference
	on Autonomous Agents and Multiagent Systems (AAMAS 2024)}{May 6 -- 10, 2024}
{Auckland, New Zealand}{N.~Alechina, V.~Dignum, M.~Dastani, J.S.~Sichman (eds.)}
\copyrightyear{2024}
\acmYear{2024}
\acmDOI{}
\acmPrice{}
\acmISBN{}

\acmSubmissionID{88}

\title[Hyper Strategy Logic]{Hyper Strategy Logic}

\author{Raven Beutner}
\orcid{0000-0001-6234-5651}
\affiliation{
  \institution{CISPA Helmholtz Center for\\Information Security}
  \country{Germany}
}

\author{Bernd Finkbeiner}
\orcid{0000-0002-4280-8441}
\affiliation{
  \institution{CISPA Helmholtz Center for\\Information Security}
  \country{Germany}}

\begin{abstract}
	Strategy logic (\SL) is a powerful temporal logic that enables strategic reasoning in multi-agent systems.
	\SL{} supports explicit (first-order) quantification over strategies and provides a logical framework to express many important properties such as Nash equilibria, dominant strategies, etc.
	While in \SL{} the same strategy can be used in multiple strategy profiles, each such profile is evaluated w.r.t.~a path-property, i.e., a property that considers the \emph{single} path resulting from a particular strategic interaction. 
	In this paper, we present Hyper Strategy Logic (\HyperSL{}), a strategy logic where the outcome of multiple strategy profiles can be compared w.r.t.~a \emph{hyperproperty}, i.e., a property that relates \emph{multiple} paths. 
	We show that \HyperSL{} can capture important properties that cannot be expressed in \SL{}, including non-interference, quantitative Nash equilibria, optimal adversarial planning, and reasoning under imperfect information.
	On the algorithmic side, we identify an expressive fragment of \HyperSL{} with decidable model checking and present a model-checking algorithm.
	We contribute a prototype implementation of our algorithm and report on encouraging experimental results.
\end{abstract}


\newcommand{\BibTeX}{\rm B\kern-.05em{\sc i\kern-.025em b}\kern-.08em\TeX}

\makeatletter
\gdef\@copyrightpermission{
	\begin{minipage}{0.3\columnwidth}
		\href{https://creativecommons.org/licenses/by/4.0/}{\includegraphics[width=0.90\textwidth]{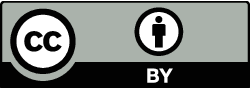}}
	\end{minipage}\hfill
	\begin{minipage}{0.7\columnwidth}
		\href{https://creativecommons.org/licenses/by/4.0/}{This work is licensed under a Creative Commons Attribution International 4.0 License.}
	\end{minipage}
	\vspace{5pt}
}
\makeatother

\begin{document}

\pagestyle{fancy}
\fancyhead{}

\maketitle

\section{Introduction}\label{sec:intro}

Two important developments in the area of reactive systems concern the study of strategic properties in multi-agent systems (MAS) and the study of hyperproperties. 
\emph{Strategic properties} analyze the ability of agents to achieve a goal against (or in cooperation) with other agents.
Logics such as alternating-time temporal logic (\ATLS{}) \cite{AlurHK02} and strategy logic (\SL{}) \cite{ChatterjeeHP10,MogaveroMPV14} reason about the temporal interaction of such agents and allow for rigorous correctness guarantees using techniques such as model-checking.
\emph{Hyperproperties} \cite{ClarksonS08} are properties that relate \emph{multiple} executions within a system. 
Hyperproperties occur in many situations in computer science where traditional path properties (that refer to \emph{individual} system execution) are not sufficient. 
Typical examples include \textbf{(1)} \emph{optimality}, e.g., one path reaching a goal faster than all other paths; \textbf{(2)} \emph{information-flow policies}, e.g., requiring that any two paths with identical low-security input should produce the same low-security output \cite{McCullough88}; and \textbf{(3)} \emph{robustness}, i.e., stating that similar inputs should lead to similar outputs \cite{ChaudhuriGL12}.

Such hyperproperties are also of vital importance in MASs. 
For example, we might ask if some agent has a strategy to achieve a goal without leaking information (an information-flow property) or can achieve a goal faster than some other agent (an optimality requirement).
Yet existing logics for strategic reasoning (such as variants of \SL{} \cite{MogaveroMPV14,ChatterjeeHP10}) cannot express such hyper-requirements (we discuss related approaches in \Cref{sec:relatedWork}).
We illustrate this on the example of Nash equilibria:

Assume we are given a MAS with agents $\{\agent{1}, \ldots, \agent{n}\}$ and LTL properties $\psi_\agent{1}, \ldots, \psi_\agent{n}$ that describe the objectives of the agents.
Agent $\agent{i}$ wants to make sure that $\ltlF \psi_\agent{i}$ holds, i.e., formula $\psi_\agent{i}$ \emph{eventually} holds.
We want to check whether the system admits a Nash equilibrium, i.e., there exists a strategy for each agent such that no agent has an incentive to deviate in order to fulfill her objective \cite{nash1950equilibrium}.
In \SL{}, we can express the existence of a Nash equilibrium as follows:
	\begin{align*}
		\exists x_1, \ldots, x_n\ldot \forall y \ldot \bigwedge_{\agent{i} = 1}^n  &\Big((\ltlF \psi_\agent{i})(\vec{x}[\agent{i} \mapsto y]) \rightarrow (\ltlF \psi_\agent{i})(\vec{x}) \Big)
	\end{align*}
where we abbreviate the strategy profiles $\vec{x} = (x_1, \ldots, x_n)$ and $\vec{x}[\agent{i} \mapsto y] = (x_1, \ldots, x_{i-1}, y, x_{i+1}, \ldots,  x_n)$.
In the variant \SL{} we consider here (similar to the \SL{} by Chatterjee et al.~\cite{ChatterjeeHP10}), atomic formulas have the form $\psi(\bX)$ where $\psi$ is an LTL formula, and $\bX$ is a strategy profile that assigns a strategy to each agent.
Formula $\psi(\vec{x})$ holds if the unique path that results from the interaction of the strategies in $\vec{x}$ satisfies $\psi$.
The above formula thus states that if some agent $\agent{i}$ can achieve $\ltlF \psi_i$  by playing some deviating strategy $y$ instead of $x_i$, i.e., the unique play that results from strategy profile $\vec{x}[\agent{i} \mapsto y]$ satisfies $\ltlF \psi_i$, then $\agent{i}$ could stick with strategy $x_i$, i.e., $\ltlF \psi_i$ also holds under strategy profile $\vec{x}$.

In the formula, we effectively compare two plays under strategy profiles $\vec{x}$ and $\vec{x}[\agent{i} \mapsto y]$.
However, \SL{} limits the comparison between multiple interactions to a boolean combination of LTL properties on their outcomes (paths).
In  game-theoretic terms, the above formula assumes that the reward for each agent is binary; the reward of agent $\agent{i}$ is maximal if $\ltlF \psi_\agent{i}$ holds and minimal if it does not. 
This fails to capture quantitative reward, for example, in a setting where agent $\agent{i}$ receives a higher reward (and thus deviates) by fulfilling $\psi_\agent{i}$ \emph{sooner}.
To express the existence of such a quantitative equilibrium, a  boolean formula over individual temporal properties on strategy profiles $\vec{x}$ and $\vec{x}[\agent{i} \mapsto y]$ is not sufficient. 
We need a more powerful mechanism that can compare the temporal behavior of \emph{multiple} paths: a \emph{hyperproperty}.

\paragraph{\HyperSL{}}

In this paper, we propose \HyperSL{} -- a new temporal logic that combines first-order strategic reasoning (as in \SL{}) with the ability to compare \emph{multiple} paths w.r.t.~a hyperproperty. 
Syntactically, we use path variables to refer to multiple paths at the same time (similar to existing hyperlogics such as \HyperCTLS{} \cite{ClarksonFKMRS14} and \HyperATLS{} \cite{BeutnerF21,BeutnerF23}).
In \HyperSL{}, atomic formulas have the form $\psi[\pi_1 : \bX_1, \ldots, \pi_m : \bX_m]$ where $\pi_1, \ldots, \pi_m$ are path variables, $\bX_1, \ldots, \bX_m$ are strategy profiles (assigning a strategy to each agent), and $\psi$ is an LTL formula where atomic propositions are indexed by path variables from $\pi_1, \ldots, \pi_m$.
The formula states that the plays resulting from strategy profiles $\bX_1, \ldots, \bX_m$, when bound to $\pi_1, \ldots, \pi_m$, (together) satisfy the hyperproperty expressed by $\psi$. 

Coming back to the Nash equilibrium example from before, we can use \HyperSL{} to express the existence of a Nash equilibrium in a quantitative reward setting as follows:
		\begin{align*}
			\exists x_1, \ldots, x_n\ldot \forall y\ldot \bigwedge_{\agent{i} = 1}^n &\Big(   \big( \neg {\psi_\agent{i}}_{\pi_1} \ltlW {\psi_\agent{i}}_{\pi_2} \big) \left[\begin{aligned}
				\pi_1 &: \bX[\agent{i} \mapsto y]\\[-1mm]
				\pi_2 &: \bX
			\end{aligned}\right]\Big)
		\end{align*}
Here, we write ${\psi_\agent{i}}_{\pi_1}$ (resp.~${\psi_\agent{i}}_{\pi_2}$) to state that $\psi_\agent{i}$ holds on path $\pi_1$ (resp.~$\pi_2$).
In the formula, we again quantify over a deviating strategy $y$, but can compare the two paths resulting from strategy profiles $\bX[\agent{i} \mapsto y]$ and $\bX$ within the \emph{same} temporal formula. 
This formula states that path $\pi_1$ (constructed using strategy profile $\bX[\agent{i} \mapsto y]$) does not satisfy $\psi_\agent{i}$ strictly before $\psi_\agent{i}$ holds on path $\pi_2$ (constructed using strategy profile $\bX$).\footnote{We make use of LTL's weak until operator $\ltlW$. Formula $\psi_1 \ltlW \psi_2$ holds if $\psi_1$ holds until $\psi_2$ holds eventually \emph{or} $\psi_1$ holds at all times.  }
If the above formula holds, $\bX$ thus constitutes a strategy profile such that no agent could achieve its goal strictly sooner (if at all). %

Note that we can express any Nash equilibrium as long as ``agent $\agent{i}$ (strictly) prefers the outcome on path $\pi_1$ over that on path $\pi_2$'' is expressible using an LTL formula over $\pi_1, \pi_2$.
Likewise, \HyperSL{} can, e.g., express that some strategy \textbf{(1)} reaches a goal without leaking information, \textbf{(2)} is at least as  fast as any other strategy, or \textbf{(3)} is robust w.r.t. the behavior of other agents.

\paragraph{Expressiveness of \HyperSL{}}

After we introduce \HyperSL{} (in \Cref{sec:hypersl}), we study its relation to existing logics (in \Cref{sec:relationToLogics}). 
We show that \HyperSL{} subsumes many non-hyper strategy logics as well as hyperlogics such as \HyperCTLS{} \cite{ClarksonFKMRS14}, \HyperATLS{} \cite{BeutnerF21,BeutnerF23}, and \HyperATLSS{} \cite{BeutnerF24} (see \Cref{sec:relatedWork}). 
Moreover, \HyperSL{} also admits reasoning under imperfect information despite having a semantics defined under complete information. 
The key observation here is that ``acting under imperfect information'' \emph{is a hyperproperty}: a strategy acts under imperfect information if, on any \emph{pair} of paths with the same observation, the strategy chooses the same action.
Formally, we show that \HyperSL{} subsumes \SLI{} \cite{BerthonMMRV17,BerthonMMRV21}, a strategy logic centered around imperfect information.

\paragraph{Model Checking}

\HyperSL{}'s ability to compare multiple strategic interactions renders model-checking (MC) undecidable. 
In \Cref{sec:decideable}, we identify a fragment of our logic -- called \HyperSLP{} -- for which MC is possible.
Intuitively, in \HyperSLP{}, the quantifier prefix should be such that we can group it into individual ``blocks'' where the strategy variables from each block are used on independent path variables. 
\HyperSLP{} subsumes \SL[1G] (the single-goal fragment of \SL{}) \cite{MogaveroMS13}, \HyperLTL{} \cite{ClarksonFKMRS14}, \HyperATLS{} \cite{BeutnerF21,BeutnerF23}, and \HyperATLSS{} \cite{BeutnerF24}, but also captures properties that cannot be expressed in existing logics.
We argue that \HyperSLP{} is the largest fragment with a decidable model-checking problem that is defined purely in terms of the quantification structure.

\paragraph{Implementation and Experiments}

We implement our MC algorithm for \HyperSLP{} in the \tool{} tool \cite{BeutnerF24} and experiment with various MAS models (in \Cref{sec:imp}).
Our experiments show that \tool{} performs well on many \emph{non}-hyper strategy logic specifications and can verify complex hyperproperties that cannot be expressed in any existing logic.

\section{Related Work}\label{sec:relatedWork}

\SL{} has been extended along multiple dimensions, including agent-unbinding \cite{LaroussinieM15}, reasoning about probabilities \cite{AminofKMMR19}, epistemic properties \cite{BelardinelliKLM21,MaubertM18,BelardinelliLMR17}, and quantitative properties \cite{BouyerKMMMP19}.
We refer to \cite{MogaveroMPV14,PaulyP03a} for a more in-depth discussion. 
The common thread in all the previous extensions is a focus on the temporal behavior on \emph{individual} paths.
\HyperSL{} generalizes \SL{} and is the first to compare \emph{multiple} paths.
Even quantitative extensions like SL[$\mathcal{F}$] \cite{BouyerKMMMP19} evaluate an \LTL[$\mathcal{F}$]-formula on a \emph{per-path} basis.
In contrast, \HyperSL{} can express complex relationships \emph{between} paths.

Studying logics that can express strategic properties under \emph{imperfect information} has attracted much attention and led to various extensions of \ATLS{} \cite{DimaT11,BullingJ14,LaroussinieMS15,BelardinelliLMR17,BerthonMM17} and \SL{}  \cite{BerthonMMRV17,abs-1908-02488}.
Berthon et al.~\cite{BerthonMMRV17} showed that their logic, \SLI{}, subsumes most existing approaches. 
We show that \HyperSL{} can also reason about imperfect information (and subsumes \SLI{}) despite having a semantics that is defined under full information.

Logics for expressing hyperproperties in non-agent-based systems (e.g., labeled transition systems) have been obtained by extending existing temporal or first-order logics with explicit path quantification over path/trace variables or an equal-level predicate \cite{ClarksonFKMRS14,Finkbeiner017,GutsfeldMO20,CoenenFHH19,BeutnerF22}.
As strategic reasoning is significantly more powerful than pure path quantification, \HyperSL{} subsumes \HyperCTLS{} (when interpreting transition systems as single-agent MASs).
\HyperATLS{} \cite{BeutnerF21,BeutnerF23} and \HyperATLSS{} \cite{BeutnerF24} extend alternating-time temporal logic (\ATLS{}) \cite{AlurHK02} with path variables and strategy-sharing constraints, leading to a strategic hyperlogic that can express important security properties such as non-deducibility of strategies \cite{WittboldJ90} and simulation security \cite{Sabelfeld03}.
Similar to \ATLS{}, the strategic reasoning in \HyperATLS{} and \HyperATLSS{} is limited to \emph{implicit} reasoning about the strategic ability of coalitions of agents and cannot explicitly reason about strategies as, e.g., needed to express the existence of a Nash equilibrium.

Our model-checking algorithm for \HyperSLP{} is based on an iterative elimination of path (variables) in an automaton, similar to existing algorithms for \HyperCTLS{} \cite{FinkbeinerRS15} and \HyperATLS{} \cite{BeutnerF23,BeutnerF24}. 
Compared to \HyperATLS{}, we need to eliminate paths by simulating an \emph{arbitrary} prefix of strategy quantifiers, leading to a more involved construction and more complex correctness proof.

\section{Preliminaries}\label{sec:prelim}
We let $\ap$ be a fixed finite set of atomic propositions and fix a fixed finite set of agents $\agents = \{\agent{1}, \ldots, \agent{n}\}$.
Given a set $X$, we write $X^+$ (resp.~$X^\omega$) for the set of non-empty finite (resp.~ infinite) sequences over $X$.
For $u \in X^\omega$ and $j \in \nat$, we write $x(j)$ for the $i$th element, $u[0,j]$ for the finite prefix up to position $j$ (of length $j + 1$), and $u[j, \infty]$ for the infinite suffix starting at position $j$.

\paragraph{Concurrent Game Structures}

As the underlying model of MASs, we use concurrent game structures (CGS) \cite{AlurHK02}.
A CGS is a tuple $\calG = (S, s_0, \moves, \kappa, L)$ where $S$ is a finite set of states, $s_0 \in S$ is an initial state, $\moves$ is a finite set of actions, $\kappa : S \times (\agents \to \moves)\to S$ is a transition function, and $L : S \to 2^\ap$ is a labeling function.
The transition function takes a state $s$ and an \emph{action profile} $\vec{a} : \agents \to \moves$ (mapping each agent an action) and returns a unique successor state $\kappa(s, \vec{a})$.
We write $\prod_{\agent{i} \in \agents} a_\agent{i}$ for the action profile where each agent $\agent{i}$ is assigned action $a_\agent{i}$.

A strategy in $\calG$ is a function $f : S^+ \to \moves$, mapping finite plays to actions.
We denote the set of all strategies in $\calG$ with $\strats{\calG}$.
A \emph{strategy profile} $\prod_{\agent{i} \in \agents} f_\agent{i}$ assigns each agent $\agent{i}$ a strategy $f_\agent{i} \in \strats{\calG}$.
Given strategy profile $\prod_{\agent{i} \in \agents} f_\agent{i}$ and state $s \in S$, we can define the unique path resulting from the interaction between the agents:
We define $\play_\calG(s, \prod_{\agent{i} \in \agents} f_\agent{i})$ as the unique path $p \in S^\omega$ such that $p(0) = s$ and for every $j \in \nat$ we have $p(j+1) = \kappa\big(p(j), \prod_{\agent{i} \in \agents} f_\agent{i}(p[0,j]) \big)$.
That is, in every step, we construct the action profile $\prod_{\agent{i} \in \agents} f_\agent{i}(p[0,j])$ in which each agent $\agent{i}$ plays the action determined by $f_\agent{i}$ on the current prefix $p[0,j]$.

\paragraph{Alternating Automata}

Our model-checking algorithm is based on alternating automata over infinite words.
These automata generalize nondeterministic automata  by alternating between nondeterministic and universal transitions \cite{Vardi95}.
For transitions of the former kind, we can choose \emph{some} successor state; for transitions of the latter type, we need to consider \emph{all} possible successor states. 
Formally, an \emph{alternating parity automaton (APA)} over alphabet $\Sigma$ is a tuple $\calA = (Q, q_0, \delta, c)$ where $Q$ is a finite set of states, $q_0 \in Q$ is an initial state, $c : Q \to \nat$ is a color assignment, and $\delta : Q \times \Sigma \to \bool^+(Q)$ is a transition function that maps each state-letter pair to a positive boolean formula over $Q$ (denoted with $\bool^+(Q)$).
For example, if $\delta(q, l) = q_1 \lor (q_2 \land q_3)$, we can -- from state $q \in Q$ and upon reading letter $l \in \Sigma$ -- either move to state $q_1$ or move to \emph{both} $q_2$ and $q_3$ (i.e., we spawn two copies of our automaton, one starting in state $q_2$ and one in $q_3$).
We write $\calL(\calA) \subseteq \Sigma^\omega$ for the set of all infinite words for which we can construct a run tree that respects the transition formulas such that the \emph{minimal} color that occurs infinitely many times (as given by $c$) is \emph{even}. 
For space reasons, we cannot give a formal semantics of APA runs and instead refer the reader to \ifFull{\Cref{app:prelim}}{the full version \cite{fullVersion}}. 
No specific knowledge about APAs is required to understand the high-level idea of our algorithm.

 A special kind of APAs are \emph{deterministic parity automata (DPA)} in which $\delta$ is a function $Q \times \Sigma \to Q$ assigning a unique successor state to each state-letter pair.
We can always  determinize APAs:
 
 \begin{proposition}[\cite{MiyanoH84,Piterman07}]\label{prop:toDPA}
 	For any APA $\calA$ with $n$ states, we can effectively compute a DPA $\calA'$ with at most $2^{2^{\calO(n)}}$ states such that $\calL(\calA) = \calL(\calA')$.
 \end{proposition}

\section{Hyper Strategy Logic}\label{sec:hypersl}

Our new logic \HyperSL{} is centered around the idea of combining strategic reasoning (as possible in strategy logic \cite{ChatterjeeHP10,MogaveroMPV14}) with the ability to express hyperproperties (as possible in logics such as 
\HyperCTLS{} \cite{ClarksonFKMRS14}).
To accomplish this, we combine the ideas from both disciples. 
On the strategy-logic-side, we use strategy variables to quantify over strategies.
On the hyper-side, we use path variables to compare multiple paths within a temporal formula. 

Let $\stratVars$ be a set of \emph{strategy variables} and $\pathVars$ a set of \emph{path variables}.
We typically use lowercase letters ($x, y, z, x_1,  \ldots$) for strategy variables and variations of $\pi$ ($\pi, \pi', \pi_1, \ldots$) for path variables.
Path and state formulas in \HyperSL{} are generated by the following grammar:
\begin{align*}
	\psi &:= a_\pi \mid \varphi_\pi \mid \psi \land \psi \mid \neg \psi \mid \ltlN \psi \mid \psi \ltlU \psi \\
	\varphi &:= \forall x. \varphi \mid \exists x. \varphi \mid  \psi\big[\pi_1 : \bX_1, \ldots, \pi_m : \bX_m \big]
\end{align*}
where $a \in \ap$ is an atomic proposition, $\pi, \pi_1, \ldots, \pi_m \in \pathVars$ are path variables, $x \in \stratVars$ is a strategy variable, and $\bX_1, \ldots, \bX_m : \agents \to \stratVars$ are \emph{strategy profiles} that assign a strategy variable to each agent.
We often write $\psi[\pi_k : \bX_k ]_{k=1}^m$ as a shorthand for $\psi[\pi_1 : \bX_1, \ldots, \pi_m : \bX_m]$. 
We use $\quant$ as a placeholder for either $\forall$ or $\exists$. 
We use the standard Boolean connectives $\lor, \to, \leftrightarrow$, and constants $\top, \bot$, as well as the derived LTL operators \emph{eventually} $\ltlF \psi := \top \ltlU \psi$ and \emph{globally} $\ltlG \psi :=\neg \ltlF \neg \psi$.
For each formula $\psi[\pi_k : \bX_k ]_{k=1}^m$, we assume that all path variables that are free in $\psi$ belong to $\{\pi_1, \ldots, \pi_m\}$, i.e., all used path variables are bound to some strategy profile.
We further assume that all nested state formulas are closed.

Note that our syntax does not support boolean combinations of state formulas as is usual in \SL{} \cite{MogaveroMPV14}. 
As we can evaluate a path formula on multiple paths, we can move boolean combinations within the path formulas.

\begin{example}\label{ex:state-conjunctions}
	Consider the \SL{} formula $\exists x\ldot (\exists y. (\ltlF a)(x, y)) \land  (\forall z. \allowbreak (\ltlG b)(z, x))$, which can be expressed in \HyperSL{} as follows: $\exists x\ldot \exists y\ldot \forall y\ldot \allowbreak (\ltlF a_{\pi_1} \land \ltlG b_{\pi_2})[\pi_1 : (x, y), \pi_2 : (z, x)]$.\demo
\end{example}

\subsubsection*{Semantics}

We fix a game structure $\calG = (S, s_0, \moves, \kappa, L)$.
A \emph{strategy assignment} is a partial mapping $\Delta : \stratVars \rightharpoonup \strats{\calG}$.
We write $\{\}$ for the unique strategy assignment with an empty domain.
In \HyperSL{}, a path formula $\psi$ refers to propositions on multiple path variables. 
We evaluate it in the context of a \emph{path assignment} $\Pi : \pathVars \rightharpoonup S^\omega$ mapping path variables to paths (similar to the semantics of \HyperCTLS{} \cite{ClarksonFKMRS14}).
Given $j \in \nat$, we define $\Pi[j, \infty]$ as the shifted assignment defined by $\Pi[j, \infty](\pi) := \Pi(\pi)[j, \infty]$.
For a path formula $\psi$, we then define the semantics in the context of path assignment $\Pi$:
\begin{align*}
	\Pi &\models_\calG a_\pi &\text{iff } \quad&a \in L\big(\Pi(\pi)(0)\big)\\
	\Pi &\models_\calG\varphi_\pi &\text{iff } \quad&\Pi(\pi)(0), \{\} \models_\calG\varphi\\
	\Pi &\models_\calG \psi_1 \land \psi_2 &\text{iff } \quad &\Pi \models_\calG \psi_1 \text{ and } \Pi \models_\calG \psi_2\\
	\Pi &\models_\calG \neg \psi &\text{iff } \quad &\Pi \not\models_\calG \psi\\
	\Pi &\models_\calG \ltlN \psi &\text{iff } \quad &\Pi[1, \infty] \models_\calG \psi\\
	\Pi &\models_\calG \psi_1 \ltlU \psi_2 &\text{iff } \quad &\exists j \in \nat\ldot \Pi[j, \infty] \models_\calG \psi_2 \text{ and } \\
	&  \quad\quad\quad\quad\quad\quad\quad\quad\quad\quad \forall 0 \leq k < j\ldot \Pi[k, \infty] \models_\calG \psi_1 \span \span
\end{align*}
The semantics for path formulas synchronously steps through all paths in $\Pi$ and evaluate $a_\pi$ on the path bound to $\pi$.
State formulas are evaluated in a state $s \in S$ and strategy assignment $\Delta$ as follows:
\begin{align*}
	&s, \Delta \models_\calG \forall x \ldot\varphi  &&\text{iff}&&\forall f \in\strats{\calG} \ldot s, \Delta[x \mapsto f] \models_\calG \varphi\\
	&s, \Delta \models_\calG \exists x \ldot\varphi  &&\text{iff} &&\exists f \in\strats{\calG} \ldot s, \Delta[x \mapsto f] \models_\calG \varphi\\
	&s, \Delta \models_\calG \psi\big[\pi_k : \bX_k \big]_{k=1}^m &&\text{iff} &&\\
	& \quad && \hspace{-5mm}\Big[\pi_k\! \mapsto \! \play_\calG\Big(s,\! \prod_{\agent{i} \in \agents} \!\!\Delta(\bX_k(\agent{i})) \Big) \Big]_{k=1}^m \models_\calG \psi \span \span \span
\end{align*}
To resolve a formula $\psi\big[\pi_k : \bX_k \big]_{k=1}^m$, we construct $m$ paths (bound to $\pi_1, \ldots, \pi_m$), and evaluate $\psi$ in the resulting path assignment. 
The $k$th path (bound to $\pi_k$) is the play where each agent $\agent{i}$ plays strategy $\Delta(\bX_k(\agent{i}))$, i.e., the strategy currently bound to the strategy variable $\bX_k(\agent{i})$.
We write $\calG \models \varphi$ if $s_0, \{\} \models_\calG \varphi$, i.e., the initial state satisfies state formula $\varphi$.

\section{Expressiveness of \HyperSL{}}\label{sec:relationToLogics}

The ability to compare multiple paths within a temporal formula makes \HyperSL{} a powerful formalism that subsumes many existing logics.
We only briefly mention some connections to existing logics.
More details can be found in \ifFull{\Cref{app:relationToLogics}}{the full version \cite{fullVersion}}.

\subsection{\SL{} and \HyperSL{}}\label{sec:translation_sl_into_hypersl}

\HyperSL{} naturally subsumes many (non-hyper) strategy logics \cite{MogaveroMPV14,ChatterjeeHP10}, which evaluate temporal properties on \emph{individual} paths.  
We consider \SL{} formulas defined by the following grammar:
\begin{align*}
	\psiSL &:= a \mid \varphi \mid \neg \psiSL \mid \psiSL \land \psiSL\mid \ltlN \psiSL \mid \psiSL \ltlU \psiSL  \\
	\varphiSL &:= \psiSL \mid \varphiSL \land \varphiSL  \mid \varphiSL \lor \varphiSL \mid \forall x\ldot \varphiSL \mid \exists x\ldot \varphiSL \mid (\agent{i}, x) \varphiSL
\end{align*}
where $a \in \ap$, $x \in \stratVars$, and $\agent{i} \in \agents$.
We assume that nested state formulas are closed.
In this \SL{}, we can quantify over strategies and \emph{bind} a strategy $x$ to agent $\agent{i}$ using $(\agent{i}, x)$; see \ifFull{\Cref{app:sl}}{the full version \cite{fullVersion}} for the full semantics. We can show the following:

\begin{restatable}{lemma}{reSlToHyper}\label{lem:sl-in-hypersl}
		For any \SL{} formula $\varphi$ there exists a \HyperSL{} formula $\varphi'$ such that for any CGS $\calG$, $\calG \models_{\text{SL}} \varphi$ iff $\calG \models \varphi'$.
\end{restatable}
\begin{proof}[Proof Sketch]
	We use a unique path variable $\dot{\pi}$.
	 During translation, we track the current strategy (variable) for each agent and construct $\dot{\pi}$ using the resulting strategy profile.
\end{proof}

\begin{example}
	Consider the formula $\exists x\ldot \forall y\ldot (\agent{1}, x) (\agent{2}, y) (\agent{3}, y) \ltlG \ltlF a$.
	We can express this formula in \HyperSL{} as $\exists x\ldot \exists y\ldot \big(\ltlG \ltlF a_{\dot{\pi}}\big)[\dot{\pi} : (x, y, y)]$ where $(x, y, y)$ denotes the strategy profile mapping agent $\agent{1}$ to $x$, and agents $\agent{2}$ and $\agent{3}$ to $y$.	\demo
\end{example}

\subsection{\HyperATLS{} and \HyperSL{}}\label{sec:translation_hyperatl_into_hypersl}

Compared to \SL{}, \ATLS{} \cite{AlurHK02} offers a weaker (implicit) form of strategic reasoning.
The \ATLS{} formula $\llangle A \rrangle \psi$ expresses that the agents in $A \subseteq \agents$ have a joint strategy to ensure path formula $\psi$ \cite{AlurHK02}. 
\HyperATLS{} \cite{BeutnerF21,BeutnerF23} is an extension of \ATLS{} that can express hyperproperties, generated by the following grammar:
\begin{align*}
	\psiATL &:= a_\pi \mid \neg \psiATL \mid \psiATL \land \psiATL \mid \ltlN \psiATL \mid \psiATL \ltlU \psiATL \\
	\varphiATL &:=\llangle A \rrangle \pi\ldot \varphiATL \mid \llbracket A \rrbracket \pi\ldot \varphiATL \mid \psiATL
\end{align*}
where $a \in \ap$, $\pi \in \pathVars$, and $A \subseteq \agents$. 
Formula $\llangle A \rrangle \pi\ldot \varphiATL$ states that the agents in $A$ have a strategy such that any path under that strategy, when bound to path variable $\pi$, satisfies the remaining formula $\varphiATL$.
Likewise, $\llbracket A \rrbracket \pi\ldot \varphiATL$ states that, no matter what strategy the agents in $A$ play, some compatible path, when bound to $\pi$, satisfies $\varphiATL$. 
See \ifFull{\Cref{app:hyperatl}}{the full version \cite{fullVersion}} for the full \HyperATLS{} semantics.
We can show the following:

\begin{restatable}{lemma}{hyperatlsToHyper}\label{lem:hyperatl-in-hypersl}
	For any \HyperATLS{} formula $\varphi$ there exists a \HyperSL{} formula $\varphi'$ such that for any CGS $\calG$, $\calG \models_\text{\HyperATLS{}} \varphi$ iff $\calG \models \varphi'$.
\end{restatable}
\begin{proof}[Proof Sketch]
	Similar to the translation of \ATLS{} to \SL{} \cite{MogaveroMPV14,ChatterjeeHP10}, we translate each \HyperATLS{} quantifier $\llangle A \rrangle \pi$ (resp.~$\llbracket A \rrbracket \pi$) using existential (resp.~universal) quantification over fresh strategies for all agents in $A$, followed by universal (resp.~existential) quantification over strategies for agents in $\agents \setminus A$ and use these strategies to construct path $\pi$.
\end{proof}

\begin{example}
	Consider the \HyperATLS{} formula $\llangle \{\agent{1}, \agent{2}\} \rrangle \pi_1\ldot \allowbreak\llangle \{\agent{3}\} \rrangle \allowbreak \pi_2\ldot (a_{\pi_1} \ltlU b_{\pi_2})$.
	We can express this in \HyperSL{} as $\exists x_\agent{1}, x_{\agent{2}}\ldot \forall x_{\agent{3}}\ldot \exists y_{\agent{3}} \ldot \allowbreak \forall y_{\agent{1}}, y_{\agent{2}}\ldot \allowbreak(a_{\pi_1} \ltlU b_{\pi_2})\big[\pi_1: (x_{\agent{1}}, x_{\agent{2}}, x_{\agent{3}}), \pi_2: (y_{\agent{1}}, y_{\agent{2}}, y_{\agent{3}})\big]$. \demo
\end{example}

By \Cref{lem:hyperatl-in-hypersl}, \HyperSL{} thus captures the various security hyperproperties (such as non-deducibility of strategies \cite{WittboldJ90} and simulation security \cite{Sabelfeld03}) that can be expressed in \HyperATLS{} \cite{BeutnerF21}. 
We can extend \Cref{lem:hyperatl-in-hypersl} further to also capture the strategy sharing constraints found in \HyperATLSS{} \cite{BeutnerF24}. 

\begin{restatable}{lemma}{hyperatlssToHyper}\label{lem:hyperatls-in-hypersl}
	For any \HyperATLSS{} formula $\varphi$ there exists a \HyperSL{} formula $\varphi'$ such that for any CGS $\calG$, $\calG \models_\text{\HyperATLSS{}} \varphi$ iff $\calG \models \varphi'$.
\end{restatable}

Moreover, \HyperSL{} can express properties that go well beyond the strict $\exists\forall$ and $\forall\exists$ quantifier alternation found in \HyperATLS{} and \HyperATLSS{} (as, e.g., needed for Nash equilibria).

\subsection{Imperfect Information and \HyperSL{}}\label{sec:hyperslAndSLI}

In recent years, much effort has been made to study strategic behavior under \emph{imperfect information} \cite{BelardinelliLM19,BerthonMMRV17,abs-1908-02488,BelardinelliLMR17,BerthonMM17,DimaT11}. 
In such a setting, an agent acts strategically (i.e., decides on an action based on its past experience) but only observes parts of the overall system. 
Perhaps surprisingly, \HyperSL{} is expressive enough to allow reasoning under imperfect information despite having a semantics with complete information (cf.~\Cref{sec:hypersl}).
Concretely, we consider strategy logic under imperfect information (\SLI{}), an extension of SL with imperfect information \cite{BerthonMMRV17,BerthonMMRV21} defined as follows:
\begin{align*}
	\psiSLI &:= a \mid \varphiSLI \mid \neg \psiSLI \mid \psiSLI \land \psiSLI \mid \ltlN \psiSLI \mid \psiSLI \ltlU \psiSLI \\
	 \varphiSLI &:= \psiSLI \mid \varphiSLI \land \varphiSLI \mid \varphiSLI \lor \varphiSLI \mid \forall x^o\ldot \varphiSLI \mid \exists x^o\ldot \varphiSLI  \mid (\agent{i}, x) \varphiSLI
\end{align*}
where $a \in \ap$, $x \in \stratVars$, $\agent{i} \in \agents$, and $o \in \obs$ is an \emph{observation} that gets attached to each strategy.
\SLI{} is evaluated on \emph{CGSs under partial observation}, which are pairs $(\calG, \{\sim_o\}_{o \in \obs})$ consisting of a CGS $\calG = (S, s_0, \moves, \kappa, L)$ and an observation relation $\sim_o \subseteq S \times S$ for each observation $o \in \obs$.
If $s\sim_os'$, then $s$ and $s'$ appear indistinguishable for a strategy with observation $o$. 
See \ifFull{\Cref{app:hypersl-and-sliii}}{the full version \cite{fullVersion}} for the full semantics.

We can effectively encode each MC instance of \SLI{} into an equisatisfiable \HyperSL{} instance (Note that the MAS models of \SLI{} and \HyperSL{} are different, so we cannot translate the formula directly but translate both the formula and the model).

\begin{restatable}{theorem}{sliToHyper}\label{prop:sliToHyper}
	For any \SLI{} MC instance $\big((\calG, \{\sim_o\}_{o \in \obs}), \varphi\big)$, we can effectively compute a \HyperSL{} MC instance $\big(\calG', \varphi'\big)$, such that $(\calG, \{\sim_o\}_{o \in \obs}) \models_{\text{SL}_{\text{ii}}} \varphi$ iff $\calG' \models \varphi'$.
\end{restatable}
\begin{proof}[Proof Sketch]
	The key observation is that a strategy acting under imperfect information \emph{is} a hyperproperty \cite{BozzelliMP15,BeutnerFFM23}:
	A strategy $f$ acts under observation $o \in \obs$ iff on any two finite paths under $f$ the action chosen by $f$ is the same, provided the two paths are indistinguishable w.r.t.~$\sim_o$. 
	We can extend the CGS $\calG$ so that the above is easily expressible in \HyperSL{}.
	We can then restrict quantification to strategies under an arbitrary observation and use a similar translation to the one used in \Cref{lem:sl-in-hypersl}.
\end{proof}

As model checking of \SLI{} is undecidable \cite{BerthonMMRV17}, we get:

\begin{corollary}\label{cor:undec}
	Model checking of \HyperSL{} is undecidable. 
\end{corollary}

\section{Model Checking of \HyperSL{}}\label{sec:decideable}

While \HyperSL{} MC is undecidable in general (cf.~\Cref{cor:undec}), we can identify fragments for which MC is possible.
For this, we cannot follow the approach of existing MC algorithms for (variants of) non-hyper \SL{}, which use tree automata to summarize strategies \cite{ChatterjeeHP10,MogaveroMPV14}.
For example, given an atomic state formula $\psi[\pi_k : \bX_k]_{k=1}^m$, we cannot construct a tree automaton that accepts all strategies that fulfill $\psi$. 
This automaton would need to \emph{compare} (and thus traverse) multiple paths in a tree at the same time.
Instead -- given the ``hyper'' origins of our logic -- we approach the MC problem by focusing on the interactions of its path variables and use \emph{word} automata to summarize satisfying path assignments. 

Throughout this section, we assume that all strategy variables are $\alpha$-renamed such that no variable is quantified more than once.

\subsection{\HyperSLP{}}

We call the fragment of \HyperSL{} we study in this section \HyperSLP{} -- short for \HyperSL{} with \textbf{S}ingle \textbf{P}ath \textbf{E}limination.

\begin{definition}\label{def:spe}
	A \HyperSLP{} formula has the form 
	\begin{align*}
		\varphi = \flat_1 \ldots \flat_m \ldot \psi\big[\pi_k :  \bX_k \big]_{k=1}^m,
	\end{align*}
	where $\flat_1, \ldots, \flat_m$ are blocks of strategy quantifiers and for each $1 \leq k \leq m$ and $\agent{i} \in \agents$, strategy variable $\bX_k(\agent{i})$ is quantified in $\flat_k$.
	We refer to $m$ as the \emph{block-rank} of $\varphi$.
\end{definition}

Intuitively, the definition states that we can partition the quantifier prefix into smaller blocks where the variables quantified in each block $\flat_k$ can be used to eliminate (construct) the (unique) path variable $\pi_k$.  
We will exploit this restriction during model-checking: we can eliminate each block of quantifiers incrementally:
as all strategies quantified in block $\flat_k$ are only needed for path $\pi_k$, we can ``construct'' $\pi_k$, and afterward forget about the strategies we have used.
Note that the definition of \HyperSLP{} only depends on the quantifier prefix and the path each strategy variable is used on; it does not make any assumption on the structure of $\psi$.

\begin{example}\label{ex:slp}
	Consider the following (abstract) \HyperSL{} formula, where $\psi$ is an arbitrary LTL formula over $\pi_1, \pi_2$.
	\begin{align*}
		\underbrace{\exists c\ldot}_{\flat_1} \underbrace{\exists z\ldot \forall w\ldot \exists v\ldot}_{\flat_2} \psi \left[\begin{aligned}
			\pi_1 &: (c, c, c, c)\\
			\pi_2 &: (w, z, v, v)
		\end{aligned}\right]
	\end{align*}
	This formula is a \HyperSLP{} formula:
	The first block $\flat_1$ consists of strategy variable $c$ and constructs $\pi_1$, and the second block $\flat_2$ constructs $\pi_2$.
	The block-rank of this formula is $2$.
	\demo
\end{example}

\subsection{Expressiveness Of \HyperSLP{}}\label{sec:sub:spe-exp}

Before we outline our model-checking algorithm for \HyperSLP{} formulas, we point to some (fragments of) other logics that fall within \HyperSLP{}.

\paragraph{\HyperATLS{} and \HyperSLP{}}
When translating \HyperATLS{} (or \HyperATLSS{}) formulas into \HyperSL{} (cf.~\Cref{lem:hyperatl-in-hypersl,lem:hyperatls-in-hypersl}), each quantifier $\llangle A \rrangle \pi$ (resp.~$\llbracket A \rrbracket \pi$) is replaced by a $\exists^*\forall^*$ (resp.~$\forall^*\exists^*$) block of strategy quantifiers that are used to construct $\pi$ (and only $\pi$).
The resulting formula is thus a \HyperSLP{} formula.

\paragraph{\SL[1G] and \HyperSLP{}}

\SL[1G] is a fragment of \SL{} that allows a prefix of strategy quantifier and agent bindings followed by a single path formula (with no nested agent binding) \cite{MogaveroMPV14,MogaveroMS13,MogaveroMS14,CermakLM15}.
When translating \SL[1G] into \HyperSL{}, we obtain a formula of the form $\quant_1 x_1\ldots \quant_m x_m\ldot \psi[\pi : \bX]$ (cf.~\Cref{lem:sl-in-hypersl}), which is trivially \HyperSLP{} as there is a single path variable (with block-rank $1$).

\paragraph{Beyond \HyperATLSS{} and \SL[1G]}
Additionally, \HyperSLP{} captures interesting hyperproperties that could not be captured in existing logics:

	\begin{example}\label{ex:planning}
	Assume a MAS with $\agents = \{\agent{r}, \agent{a}, \agent{\mathit{ndet}}\}$ describing a planning task between a robot $\agent{r}$ that wants to reach a state where AP $\mathit{goal} \in \ap$ holds, and an adversary $\agent{a}$ that wants to prevent the robot from reaching the goal. 
	In each step, agent $\agent{r}$ can select a direction to move in, and $\agent{a}$ can choose a direction it wants to push the robot to. 
	Each combination of actions of $\agent{r}$ and $\agent{a}$ results in a set of potential successor locations, and the nondeterminism agent $\agent{\mathit{ndet}}$ decides which of those locations the robot actually moves to. 
	We want to check if agent $\agent{r}$ has a winning strategy that \emph{can} reach the goal against all possible behaviors of agent $\agent{a}$, i.e., $\agent{r}$ needs to reach the goal under favorable non-deterministic outcomes. 
	We can express this (non-hyper) property in \HyperSLP{} as
	\begin{align*}
		\exists x\ldot \forall y \ldot \exists z\ldot \big(\ltlF \mathit{goal}_\pi\big) \big[ \pi : (x, y, z) \big],
	\end{align*}
where we write $(x, y, z)$ for the strategy profile that assigns agent $\agent{r}$ to $x$, agent $\agent{a}$ to $y$, and agent $\agent{\mathit{ndet}}$ to $z$.
	In \HyperSLP{}, we can additionally state that $\agent{r}$ should reach the goal \emph{as fast as possible}, i.e., at least as fast as any path in the MAS:
	\begin{align*}
		\exists x\ldot \forall y \ldot \exists z\ldot \forall a\ldot \forall b\ldot \forall c\ldot (\neg \mathit{goal}_{\pi'}) \ltlU \mathit{goal}_\pi \left[\begin{aligned}
			&\pi : (x, y, z)\\
			&\pi' : (a, b, c)
		\end{aligned}\right]
	\end{align*}
	Here, we quantify over any potential different path $\pi'$ and state that $\pi$ is at least as fast as $\pi'$.
	Such requirements cannot be expressed in \SL{} (even in quantitative versions like \SL[$\mathcal{F}$]), nor can they be expressed in \HyperATLS{} or \HyperATLSS{}. \demo
\end{example}

\subsection{Summarizing Path Assignments}

In the remainder of this section, we prove the following:

\begin{theorem}\label{theo:dec}
	Model checking for \HyperSL[SPE] is decidable.
\end{theorem}

We fix a CGS $\calG = (S, s_0, \moves, \kappa, L)$ and state $\dot{s} \in S$, and let $\varphi = \flat_1 \ldots \flat_m \ldot \psi\big[\pi_k :  \bX_k \big]_{k=1}^m$ be a \HyperSLP{} formula.
We want to check if $\dot{s}, \{\} \models_\calG \varphi$, i.e., $\varphi$ holds in state $\dot{s}$.

\paragraph{Zipping Path Assignments}

The main idea of our algorithm is to summarize path assignments that satisfy subformulas of $\varphi$, similar to MC algorithms for \HyperLTL{}, \HyperCTLS{}, and \HyperATLS{} \cite{FinkbeinerRS15,BeutnerF23b,BeutnerF23,BeutnerF24}.
To enable automata-based reasoning about path assignments,  i.e., mappings $\Pi : V \to S^\omega$ for some $V \subseteq \pathVars$, we \emph{zip} such an assignment into an infinite word.
Concretely, given $\Pi : V \to S^\omega$ we define $\mathit{zip}(\Pi) \in (V \to S)^\omega$ as the infinite word of functions where $\mathit{zip}(\Pi)(j)(\pi) := \Pi(\pi)(j)$ for every $j \in \nat$, i.e., the function in the $j$th step maps each path variable $\pi \in V$ to the $j$th state on the path bound to $\pi$.

\paragraph{Summary Automaton}

Given a quantifier block $\flat = \quant_1 x_1\ldots \quant_n x_n$ over strategy variables $x_1, \ldots, x_n$, we define $\widetilde{\flat}$ as the analogous block of quantification of strategies $f_{x_1}, \ldots, f_{x_n}$, i.e., $\widetilde{\flat} := \quant_1 f_{x_1} \in \strats{\calG} \ldots \quant_n f_{x_n} \in \strats{\calG}$.
At the core of our model-checking algorithm, we construct automata that accept (zippings of) partial satisfying path assignments.
Formally:

\begin{definition}\label{def:equiv}
	For $1 \leq k \leq m+1$, we say an automaton $\calA$ over alphabet $(\{\pi_{1}, \ldots, \pi_{k-1}\} \to S)$ is a \emph{$(\calG, \dot{s}, k)$-summary} if for every path assignment $\Pi : \{\pi_{1}, \ldots, \pi_{k-1}\} \to S^\omega$ we have $\mathit{zip}(\Pi) \in \calL(\calA)$ if and only if
	\begin{align*}
		&\widetilde{\flat_k}\cdots \widetilde{\flat_m}\ldot \Pi\Big[\pi_j \mapsto \play_\calG(\dot{s}, \prod_{\agent{i} \in \agents} f_{\bX_j(\agent{i})})\Big]_{j = k}^m \models_\calG \psi.
	\end{align*}
\end{definition}

That is, a $(\calG, \dot{s}, k)$-summary accepts (the zipping of) a path assignment $\Pi$ over paths $ \pi_{1}, \ldots, \pi_{k-1}$ if -- when simulating the quantification over strategies needed to construct paths $\pi_{k}, \ldots, \pi_{m}$ and adding them to $\Pi$ -- the body $\psi$ of the formula is satisfied.

\begin{example}\label{ex:equiv}
	We illustrate the concept using the abstract formula from \Cref{ex:slp}. 
	A $(\calG, \dot{s}, 3)$-summary is an automaton $\calA_3$ over alphabet $(\{\pi_{1}, \pi_{2}\} \to S)$ such that for every $\Pi : \{\pi_1, \pi_2\} \to S^\omega$ we have $\mathit{zip}(\Pi) \in \calL(\calA_3)$ iff $\Pi \models_\calG \psi$.
	A $(\calG, \dot{s}, 2)$-summary is an automaton $\calA_2$ over alphabet $(\{\pi_{1}\} \to S)$ such that  for every $\Pi : \{\pi_1\} \to S^\omega$ we have $\mathit{zip}(\Pi) \in \calL(\calA_2)$ iff
	\begin{align*}
		\exists f_z\ldot \forall f_w \ldot \exists f_v\ldot \Pi\big[\pi_2 \mapsto \play_\calG(\dot{s}, (f_w, f_z, f_v, f_v))\big] \models_\calG \psi,
	\end{align*}
	i.e., we mimic the quantification of block $\flat_2$ to construct path $\pi_2$ (using the quantified strategies $f_z, f_w, f_v \in \strats{\calG}$) and add this path to $\Pi$ (which already contains $\pi_1$).\demo
\end{example}

\subsection{Constructing $(\calG, \dot{s}, k)$-Summaries}

\begin{algorithm}[!t]
	\caption{Simulation construction for block elimination. 
	}\label{alg:product}
\begin{code}
def simulate($\calG = (S, s_0, \moves, \kappa, L)$,$\dot{s}$,$\pi$,$\vec{x}$,$\flat = \quant_1 x_1 \ldots \quant_n x_n$,$\calA$):
@@$\calA_\mathit{det} = (Q, q_0, \delta, c)$ = toDPA($\calA$) // (*\color{comment-color}Using \Cref{prop:toDPA} \label{line:toDPA}*)
@@$\calB$  = $(Q \times S, (q_0, \dot{s}), \delta', c'\big)$ where 
@@@@$c'(q, s) := c(q)$ (*\label{line:color}*)
@@@@$\begin{aligned}
\delta'\big( (q, s) , \bT\big) := \biglandlor^{\quant_1}_{a_{x_1} \in \moves }\!\! \cdots \!\! \biglandlor^{\quant_n}_{a_{x_n} \in \moves} \bigg(\delta\big(q, \bT[\pi \mapsto s] \big), \kappa\Big(s, \prod_{\agent{i} \in \agents} a_{\bX(\agent{i})}   \Big)\bigg)
\end{aligned}$(*\label{line:tran}*)
@@return $\calB$
\end{code}
	\vspace{-1mm}
\end{algorithm}

We write $\biglandlor^\quant$ for a conjunction ($\bigwedge$) if $\quant = \forall$ and a disjunction ($\bigvee$) if $\quant = \exists$.
The backbone of our model-checking algorithm (which we present in \Cref{sec:sub:mc}) is an effective construction of a $(\calG, \dot{s}, k)$-summary $\calA_k$ for each $1 \leq k \leq m + 1$.
To construct these summaries, we simulate quantification over strategies.
We describe this simulation construction in \Cref{alg:product}.
Before explaining the construction, we state the result of \Cref{alg:product} as follows:

\begin{restatable}{proposition}{simProp}\label{prop:sim}
	Given $\dot{s} \in S$, $\pi \in \pathVars$, a strategy profile $\vec{x} : \agents \to \stratVars$, a quantifier block $\flat$ such that for every $\agent{i} \in \agents$, $\vec{x}(\agent{i})$ is quantified in $\flat$, and an APA $\calA$ over alphabet $(V \uplus \{\pi\} \to S)$.
	Let $\calB$ be the results of \lstinline[style=default, language=custom-lang]|simulate($\calG$,$\dot{s}$,$\pi$,$\bX$,$\flat$,$\calA$)|.
	Then for any path assignment $\Pi : V \to S^\omega$, we have $\mathit{zip}(\Pi) \in \calL(\calB)$  iff
	\begin{align}\label{eq:simStatement}
		&\widetilde{\flat}\ldot  \mathit{zip}\Big(\Pi\big[\pi \mapsto \play_\calG\big(\dot{s}, \prod_{\agent{i} \in \agents} f_{\bX(\agent{i})}\big)\big] \Big)\in \calL(\calA).
	\end{align}
\end{restatable}

That is, the automaton $\calB$ accepts the zipping of an assignment $\Pi : V \to S^\omega$ iff by simulating the quantifier prefix in $\flat$, we construct a path for $\pi$ that, when added to $\Pi$, is accepted by $\calA$.
Note the similarity to \Cref{def:equiv}: In \Cref{def:equiv} we simulate multiple quantifier blocks to construct paths $\pi_k, \ldots, \pi_m$ that, when added to $\Pi$, should satisfy the body $\psi$. 
In \Cref{prop:sim}, we simulate a single path that, when added to $\Pi$, should be accepted by automaton $\calA$. 
We will later use \Cref{prop:sim} to simulate one quantifier block at a time, eventually reaching an automaton required by \Cref{def:equiv}.

Before proving \Cref{prop:sim}, let us explain the automaton construction in \lstinline[style=default, language=custom-lang]|simulate| (\Cref{alg:product}).
In Eq.~(\ref{eq:simStatement}), $\widetilde{\flat}$ quantifies over strategies in $\calG$, which are infinite objects (function $S^+ \to \moves$).
The crucial point that we will exploit is that the underlying game the strategies operate on is \emph{positionally determined}. 
The automaton we construct can, therefore, \emph{simulate} the path $\pi$ in $\calG$ and select fresh actions in each step (instead of fixing strategies globally) \cite{BeutnerF21,BeutnerF23,BeutnerF24}.
To do this, we first translate the APA $\calA$ to a DPA $\calA_\mathit{det} = (Q, q_0, \delta, c)$ (in line \ref{line:toDPA}). 
The new automaton $\calB$ then simulates path $\pi$ by tracking its current state in $\calG$ and simultaneously tracks the current state of $\calA_\mathit{det}$, thus operating on states in $Q \times S$.
We start in state $(q_0, \dot{s})$, i.e., the initial state of $\calA_\mathit{det}$ and the designed state $\dot{s}$ from which we want to start the simulation of $\pi$.
The color of each state is simply the color of the automaton we are tracking, i.e., $c'(q, s) = c(q)$ (line \ref{line:color}).
During each transition, we then update the current state of $\calA_\mathit{det}$ and the state of the simulation (defined in line \ref{line:tran}).
Concretely, when in state $(q, s)$, we read a letter $\bT : V \to S$ that assigns states to all path variables in $V$ (recall that the alphabet of $\calA$ is $V \cup \{\pi\} \to S$ and the alphabet of $\calB$ is $V \to S$).
We update the state of $\calA_\mathit{det}$ to $\delta(q, \bT[\pi\mapsto s])$, i.e., we extend the input letter $\bT$ with the current state $s$ of the simulation of path $\pi$ (note that $\bT[\pi\mapsto s] : V \cup \{\pi\} \to S$). 
To update the simulation state $s$, we make use of the positional determinacy of the game:
Instead of quantifying over strategies (as in Eq.~(\ref{eq:simStatement})), we can quantify over actions in each step of the automaton.
Concretely, for each universally quantified strategy variable in $\flat$, we pick an action conjunctively, and for each existentially quantified variable, we pick an action disjunctively.  
After we have picked actions $a_{x_1}, \ldots, a_{x_n}$ for all strategies quantified in $\flat$, we can update the state of the $\pi$-simulation by constructing the action assignment $\prod_{\agent{i} \in \agents} a_{\vec{x}(\agent{i})}$, i.e., assign each agent the corresponding action, and obtain the next state using $\calG$'s transition function $\kappa$.

\begin{example}\label{ex:construction}
	Let us use \Cref{ex:slp} to illustrate the construction in \Cref{alg:product}.
	Assume we are given an $(\calG, \dot{s}, 3)$-summary $\calA_3$ over alphabet $(\{\pi_{1}, \pi_{2}\} \to S)$, i.e., for every $\Pi : \{\pi_1, \pi_2\} \to S^\omega$, we have $\mathit{zip}(\Pi) \in \calL(\calA_3)$ iff $\Pi \models_\calG \psi$ (cf.~\Cref{ex:equiv}).
	We invoke \lstinline[style=default, language=custom-lang]|simulate($\calG$,$\dot{s}$,$\pi_2$,$\bX$,$\flat_2$,$\calA_3$)| where $\bX = (w, z, v, v)$ and $\flat_2 = \exists z \forall w\exists v$, and let $(Q, q_0, \delta, c)$ be the DPA equivalent to $\calA_3$ (computed in line \ref{line:toDPA}). 
	In this case, \lstinline[style=default, language=custom-lang]|simulate| computes the APA $\calB = (Q \times S, (q_0, \dot{s}), \delta', c')$ over alphabet $\{\pi_1\} \to S$ where $\delta'\big((q, s), \bT\big)$ is defined as
	\begin{align*}
		\bigvee_{a_z \in \moves} \; \bigwedge_{a_w \in \moves} \; \bigvee_{a_v \in \moves}  \Big(\delta(q, \bT[\pi_2 \mapsto s]), \kappa\big(s, (a_w, a_z, a_v, a_v) \big) \Big).
	\end{align*}
	That is, in each step, we disjunctively choose an action $a_z$ (corresponding to the action selected by existentially quantified strategy $z$), conjunctively pick an action $a_w$ (corresponding to the action selected by universally quantified strategy $w$), and finally disjunctively select action $a_v$.
	After we have fixed actions $a_z$, $a_w$ and $a_v$, we take a step in $\calG$ by letting each agent $\agent{i}$ play action $a_{ \bX(\agent{i})}$, i.e., agent $\agent{1}$ chooses action $a_w$, agent $\agent{2}$ chooses $a_z$, and agents $\agent{3}$ and $\agent{4}$ pick $a_v$.
	By \Cref{prop:sim}, every $\Pi : \{\pi_1\} \to S^\omega$ satisfies $\mathit{zip}(\Pi) \in \calL(\calB)$ iff 
	\begin{align*}
		\exists f_z\ldot \forall f_w \ldot \exists f_v\ldot \mathit{zip}\big(\Pi[\pi_2 \mapsto \play_\calG(\dot{s}, (f_w, f_z, f_v, f_v))]\big) \in \calL(\calA_3)
	\end{align*}
	which (by assumption on $\calA_3$) holds iff 
	\begin{align*}
		\exists f_z\ldot \forall f_w \ldot \exists f_v\ldot \Pi[\pi_2 \mapsto \play_\calG(\dot{s}, (f_w, f_z, f_v, f_v))] \models_\calG \psi.
	\end{align*}
	We have thus used \lstinline[style=default, language=custom-lang]|simulate| (\Cref{alg:product}) to compute a $(\calG, \dot{s}, 2)$-summary from a $(\calG, \dot{s}, 3)$-summary (cf.~\Cref{ex:equiv}). \demo
\end{example}

We can now formally prove \Cref{prop:sim}:

\begin{proof}[Proof Sketch of \Cref{prop:sim}]
	The idea of automaton $\calB$  constructed in \Cref{alg:product} is to simulate the path that corresponds to path variable $\pi$.
	To argue that $\calB$ expresses the desired language, we make use of the positional determinacy of concurrent parity games (CPG) \cite{MalvoneMS16}.
	A CPG is a simple multi-player game model where we can quantify over strategies for each of the players.
	For any fixed $\Pi$, we design an (infinite-state) CPG, that is won iff Eq.~(\ref{eq:simStatement}) holds.
	We then exploit the fact that CPGs are determined (cf.~\cite[Thm.~4.1]{MalvoneMS16}), i.e., instead of quantifying over entire strategies in the CPG, we can quantify over Skolem functions for actions in each step.
	This allows us to show that the CPG is won iff $\calB$ has an accepting run (on the fixed $\Pi$), giving us the desired result. 
	We refer the interested reader to \ifFull{\Cref{app:simProof}}{the full version \cite{fullVersion}} for details.
\end{proof}

\subsection{Model-Checking Algorithm}\label{sec:sub:mc}

\begin{algorithm}[!t]
	\caption{Model-checking algorithm for \HyperSLP{}.}\label{alg:main-alg}
\begin{code}
def modelCheck($\calG$,$\dot{s}$,$\varphi =  \flat_1 \cdots \flat_m \ldot \psi \big[\pi_k : \bX_k\big]_{k=1}^m $):
@@// Assume (*\color{comment-color}$\psi$*) contains no nested state formulas
@@$\calA_{m+1}$ = LTLtoAPA($\psi$) 
@@// (*\color{black!50!white}$\calA_{m+1}$ is a $(\calG, \dot{s}, m+1)$-summary*)(*\label{line:ltlToApa}*)
@@for $k$ from $m$ to $1$:
@@@@$\calA_k$ = simulate($\calG$,$\dot{s}$,$\pi_k$,$\bX_k$,$\flat_k$,$\calA_{k+1}$) (*\label{line:product}*)
@@@@// (*\color{black!50!white}$\calA_k$ is a $(\calG, \dot{s}, k)$-summary*) (*\label{line:inv}*)
@@if $\calL(\calA_1) \neq \emptyset$ then (*\label{line:emptinessCheck}*)
@@@@return SAT // (*\color{black!50!white}$\dot{s}, \{\} \models_\calG \varphi$*)
@@else 
@@@@return UNSAT  // (*\color{black!50!white}$\dot{s}, \{\} \not\models_\calG \varphi$ *)
\end{code}
\end{algorithm}

Equipped with the concept of $(\calG, \dot{s}, k)$-summary and the simulation construction, we can now present our MC algorithm for \HyperSLP{} in \Cref{alg:main-alg}.
The \lstinline[style=default, language=custom-lang]|modelCheck| procedure is given a CGS $\calG$, a state $\dot{s}$, and a \HyperSLP{} formula $\varphi$, and checks if $\dot{s}, \{\} \models_\calG \varphi$. 
Our algorithm assumes, w.l.o.g., that the path formula $\psi$ contains no nested state formulas.
In case there are nested state formulas, we can eliminate them iteratively:
We recursively check each nested state formula on all states of the CGS, and label all states where the state formula holds with a fresh atomic proposition. 
In the path formula, we can then replace each state formula with a reference to the fresh atomic proposition. 
See, e.g., \cite{EmersonH86,BeutnerF24} for details. 

The main idea of our MC algorithm is to iteratively construct a $(\calG, \dot{s}, k)$-summary $\calA_k$ for each $1 \leq k \leq m + 1$.
Initially, in line \ref{line:ltlToApa}, we construct a $(\calG, \dot{s}, m+1)$-summary $\calA_{m+1}$ using a standard construction to translate the LTL formula $\psi$ to an APA over alphabet $(\{\pi_1, \ldots, \pi_m\} \to S)$, as is, e.g., standard for \HyperCTLS{} \cite{FinkbeinerRS15}.
For each $k$ from $m$ to $1$, we then use the $(\calG, \dot{s}, k+1)$-summary $\calA_{k+1}$ to compute a $(\calG, \dot{s}, k)$-summary $\calA_k$ using the \lstinline[style=default, language=custom-lang]|simulate| construction from \Cref{alg:product} (similar to what we illustrated in \Cref{ex:construction}).
From \Cref{prop:sim}, we can conclude the following invariant:

\begin{restatable}{lemma}{lemmaInd}\label{lem:inv}
	In line \ref{line:inv}, $\calA_k$ is a $(\calG, \dot{s}, k)$-summary.
\end{restatable}

After the loop, we are thus left with a $(\calG, \dot{s}, 1)$-summary $\calA_1$ (over the simpleton alphabet $(\emptyset \to S)$) and can check if $\dot{s}, \{\} \models_\calG \varphi$ by testing $\calA_1$ for emptiness (line \ref{line:emptinessCheck}):

\begin{restatable}{lemma}{lemmaInit}\label{lem:1summary}
	For any $(\calG, \dot{s}, 1)$-summary $\calA$, we have that $\calL(\calA) \neq \emptyset$ if and only if $\dot{s}, \{\} \models_\calG \varphi$.
\end{restatable}

From \Cref{lem:inv,lem:1summary}, it follows that \lstinline[style=default, language=custom-lang]|modelCheck($\calG$,$\dot{s}$,$\varphi$)| returns \texttt{SAT} iff $\dot{s}, \{\} \models_\calG \varphi$, proving \Cref{theo:dec}.

\subsection{Model-Checking Complexity}

The determinization in line \ref{line:toDPA} of \Cref{alg:product} results in a DPA $\calA_\mathit{det}$ of doubly exponential size (cf.~\Cref{prop:toDPA}).
The size of $\calB$ is then linear in the size of $\calA_\mathit{det}$ and $\calG$.
In the worst case, each call of \lstinline[style=default, language=custom-lang]|simulate| thus increases the size of the automaton by two exponents.
For a \HyperSLP{} formula with block-rank $m$, \lstinline[style=default, language=custom-lang]|simulate|  is called $m$ times, so the final automaton $\calA_1$ has, in the worst case, $2m$-exponential many states (in the size of $\psi$ and $\calG$).
As we can check emptiness of APAs over the singleton alphabet $(\emptyset \to S)$ in polynomial time, we get:

\begin{restatable}{theorem}{compTheo}\label{theo:comp}
	Model checking for a \HyperSLP{} formula with block-rank $m$ is in $2m$-\EXPTIME{}.
\end{restatable} 

From \Cref{lem:hyperatl-in-hypersl} and the lower bounds known for \HyperATLS{} \cite{BeutnerF23}, it follows that our algorithm is asymptotically almost optimal:

\begin{restatable}{lemma}{lowerBound}\label{prop:lowerBound}
	Model checking for a \HyperSLP{} formula with block-rank $m$ is $(2m-1)$-\EXPSPACE{}-hard.
\end{restatable} 

\begin{table}
	\caption{We compare \tool{} and \mcmassl{} on the scheduler problem from \cite{CermakLM15}. 
		We give the size of the system ($\boldsymbol{|S|}$), the size of the reachable fragment ($\boldsymbol{|S_\mathit{reach}|}$), and the times in seconds ($\boldsymbol{t}$). 
		The timeout (TO) is set to 1 h.\vspace{0mm}
	}\label{tab:results-sl}
	\centering
	\small
		\begin{tabular}{lllll}
			\toprule
			$\boldsymbol{n}$ & $\boldsymbol{|S|}$ & $\boldsymbol{|S_\mathit{reach}|}$ &  $\boldsymbol{t}_{\mcmassl{}}$& $\boldsymbol{t}_\tool{}$  \\
			\midrule
			2 &72 & 9 & \textbf{0.1} & 0.4  \\
			3 &432& 21 & 6.71 & \textbf{1.9}  \\
			4 &2592& 49 & 313.7 & \textbf{24.5}   \\
			5 &15552& 113& TO & \textbf{332.1} \\
			\bottomrule
		\end{tabular}
	
\end{table}

\subsection{Beyond \HyperSLP{}}\label{sec:beyond_spe}

\HyperSLP{} is defined purely in terms of the structure of the quantifier prefix.
As soon as strategy variables are quantified in an order such that they cannot be grouped together, MC  becomes, in general, undecidable:
Already the simplest such property $\quant x. \quant y. \quant z. \quant w\ldot \allowbreak \psi[\pi_1 : (x, z), \pi_2 : (y, w)]$, leads to undecidable MC (see \ifFull{\Cref{app:beyond}}{the full version \cite{fullVersion}}).
The fragment we have identified is thus the largest possible (when only considering the quantifier prefix).
Any further study into decidable fragments of \HyperSL{} needs to impose restrictions beyond the prefix and, e.g., analyze how different path variables are related within an LTL path formula (see also \Cref{sec:conclusion}).

\section{Implementation and Experiments}\label{sec:imp}

We have implemented our \HyperSLP{} model-checking algorithm in the \tool{} tool \cite{BeutnerF24}.

\subsection{Model-Checking For Strategy Logic}\label{sec:sub:atls}

We compare \tool{} against \mcmassl{} \cite{CermakLM15} on (non-hyper) \SL[1G] properties (cf.~\Cref{sec:sub:spe-exp}).
In  \Cref{tab:results-sl}, we depict the verification times for the scheduling problem from \cite{CermakLM15} (which can be expressed in \SL[1G] and \ATLS{}).
As in \cite{BeutnerF24}, we observe that \tool{} performs much faster than \mcmassl{}, which we largely accredit to \tool{}'s efficient automata backend using \texttt{spot} \cite{Duret-LutzRCRAS22}.
Note that we use \mcmassl{} and \tool{} directly on the original model, i.e., we did not perform any prepossessing using, e.g., abstraction techniques \cite{BallK06,BelardinelliL17,Belardinelli0JM23} (which would reduce the system size and make the verification more scaleable for both tools).

\subsection{Model-Checking For Hyperproperties}\label{sec:sub:hyper}

\begin{table}
	\caption{We check random formulas from the (\textbf{Sec}), (\textbf{GE}), and (\textbf{Rnd}) templates on the ISPL models from \cite{LomuscioQR09}.
		For each model and template, we sample 10 formulas and report the average time (in seconds). 
	}\label{tab:results-hyper}

	\centering
	\small
	\begin{tabular}{lccccc}
		\toprule
		\textbf{Model} & \textbf{Sec} & \textbf{GE} & \textbf{Rnd}$_2$ & \textbf{Rnd}$_3$ & \textbf{Rnd}$_4$\\
		\midrule
		\textsc{bit-transmission} & 0.6 & 0.7 & 0.8  & 0.8 & 2.7 \\
		\textsc{book-store}  & 0.4 & 0.4 & 0.4  & 0.5 & 0.5 \\
		\textsc{card-game}   & 0.4 & 0.5 & 0.4  & 0.5 & 0.5 \\
		\textsc{dining-cryptographers}   & 0.6 & 2.7 & 11.4  & 22.6 & 10.3 \\
		\textsc{muddy-children}  & 0.4 & 3.0 & 1.7  & 0.8 & 0.9 \\
		\textsc{simple-card-game}   & 0.3 & 3.4 & 2.9 & 25.3 & 32.6  \\
		\textsc{software-development}  & - & - & -  & - & - \\
		\textsc{strongly-connected}  & 0.6 & 0.8 & 0.8 & 1.7 & 3.2  \\
		\textsc{tianji-horse-racing} & 0.4 & 0.5 & 0.4 & 0.5 & 0.5  \\
		\bottomrule
	\end{tabular}
\end{table}

In a second experiment, we demonstrate that \tool{} can verify hyperproperties on various MASs from the literature. 
We use the ISPL models from the \mcmas{} benchmarks suit \cite{LomuscioQR09}, and generate random \HyperSLP{} formulas from various property templates:
\begin{itemize}[leftmargin=*]
	\item \textbf{(Sec):} We check if some agent $\agent{i}$ can reach some target state without leaking information about some secret AP via some observable AP. 
	Concretely, we check if $\agent{i}$ can play such that on some other path, the same observation sequence is coupled with a different high-security input, a property commonly referred to as \emph{non-inference} \cite{McLean94} or \emph{opacity} \cite{ZhangYZ19,SabooriH13}. 
	\item \textbf{(GE):} We check if a given \SL[1G] formula holds on all input sequences for which \emph{some} winning output sequence exists, as is, e.g., required in \emph{good-enough} synthesis \cite{AlmagorK20,AminofGR21}.
	\item \textbf{(Rnd):} We randomly generate \HyperSLP{} formulas with block-rank $2, 3$, and $4$ (called \textbf{Rnd}$_2$, \textbf{Rnd}$_3$, and \textbf{Rnd}$_4$, respectively).
\end{itemize}

We depict the results in \Cref{tab:results-hyper}, demonstrating that \tool{} can handle most instances.
The only exception is the \textsc{software-development} model, which includes $\approx$15k states and is therefore too large for an automata-based representation. 

We stress that we do not claim that all formulas in each of the templates model realistic properties in each of the systems. 
Rather, our evaluation \textbf{(1)} demonstrates that \HyperSLP{} can express interesting properties, and \textbf{(2)} empirically shows that \tool{} can check such properties in existing ISPL models (confirming this via further real-world scenarios is interesting future work).

\subsection{Model-Checking For Optimal Planning}\label{sec:sub:planning}

\begin{table}
	\caption{We use \tool{} to solve the optimal adversarial planning problem (cf.~\Cref{ex:planning}) for varying sizes. Times are given in seconds, and the TO is set to 120 sec. \vspace{0mm}	}\label{tab:results-planning}

	\centering
	\small
	\begin{tabular}{llllllllll}
		\toprule
		\textbf{Size} & 40 & 50 & 60 & 70 & 80 & 90 & 100 & 110 & 120\\
		\midrule
		\textbf{t} & 14.2 & 22.0 & 31.2 & 42.5& 57.6 & 70.1 & 86.8 & 104.6 & TO\\
		\bottomrule
	\end{tabular}

\end{table}

In our last experiment, we challenge \tool{} with planning examples as those outlined in \Cref{ex:planning}.
We randomly generate planning instances between the robot $\agent{r}$, adversary $\agent{a}$, and $\agent{\mathit{ndet}}$, and check if robot $\agent{r}$ \emph{can} reach the goal following some shortest path in the problem. 
For a varying size $n$, we randomly create $10$ planning instances with $n$ states.
We report the verification times in \Cref{tab:results-planning}.
With increasing size, the running time of \tool{} clearly increases, but the increase seems to be quadratic rather than exponential.

\section{Conclusion and Future Work}\label{sec:conclusion}

We have presented \HyperSL{}, a new temporal logic that extends strategy logic with the ability to reason about hyperproperties. 
\HyperSL{} can express complex properties in MASs that require a combination of strategic reasoning and hyper-requirements (such as optimalilty, GE, non-interference, and quantitative Nash equilibria); many of which were out of reach of existing logics. 
As such, \HyperSL{} can serve as a unifying foundation for an exact exploration of the interaction of strategic behavior with hyperproperties, and provides a formal language to express (un)decidability results. 
Moreover, we have taken a first step towards the ambitious goal of automatically model-checking \HyperSL{}.
Our fragment \HyperSLP{} subsumes many relevant other logics and captures unique properties not expressible in existing frameworks. 
Our implementation in \tool{} shows that our MC approach is practical in small MASs.

A particularly interesting future direction is to search for further fragments of \HyperSL{} with decidable model checking.
As argued in \Cref{sec:beyond_spe}, any such fragment needs to take the structure of the LTL-formula(s) into account. 
For example, Mogavero et al.~\cite{MogaveroMS13} showed that \SL{}[CG] (a fragment of \SL{} that only allows conjunctions of goal formulas) still admits behavioral strategies (i.e., strategies that do not depend on future or counterfactual decisions of other strategies).
When extending this to our hyper setting, it seems likely that if a strategy is used on multiple path variables, but these paths occur in disjoint conjuncts of path formulas, MC remains decidable. 
We leave such extensions as future work.

\begin{acks}
	This work was supported by the European Research Council (ERC) Grant HYPER (101055412), and by the German Research Foundation (DFG) as part of TRR 248 (389792660).
\end{acks}

\bibliographystyle{ACM-Reference-Format} 
\bibliography{references}

\iffullversion

\clearpage

\appendix

\section{Additional Material for Section \ref{sec:prelim}}\label{app:prelim}

In this section we give details on the semantics of APAs. 
To make our later proofs (which use the APA semantics) easier, we use a DAG-based semantics (opposed to the more prominently used -- but equivalent -- tree-based semantics).
We refer the reader to \cite{Vardi95} for more details.

\begin{definition}
	For a set $Q$, we write $\bool^+(Q)$ for the set of all positive boolean formulas over $Q$, i.e., all formulas generated by the grammar
	\begin{align*}
		\theta := q \mid \theta_1 \land \theta_2 \mid \theta_1 \lor \theta_2
	\end{align*}
	where $q \in Q$.
	Given a subset $X \subseteq Q$ and $\theta \in \bool^+(Q)$, we write $X \models \theta$ if the assignment that maps all states in $X$ to $\top$ and those in $Q \setminus X$ to $\bot$ satisfies $\Psi$. 
	For example $\{q_0, q_1\} \models q_0 \land (q_1 \lor q_2)$.
\end{definition}

Let $\calA = (Q, q_0, \delta, c)$ be an APA.
A run DAG of $\calA$ is a pair $\mathbb{D} = (V, E)$ of nodes and edges such that $V \subseteq Q \times \nat$, $E \subseteq \bigcup_{i \in \nat} (Q \times \{i\}) \times (Q \times \{i+1\})$.
That is, each node in $V$ consist of a state in $Q$ and a depth and the edges in $E$ only reach from depth $i$ to depth $i+1$.
For every $(q, i) \in V$, we define $\mathit{sucs}(q, i) := \{q' \mid ((q, i), (q', i+1)) \in E\}$ as the state component of $(q, i)$'s successor nodes. 

A run of $\calA$ on a word $u \in \Sigma^\omega$ is a run DAG $\mathbb{D} = (V, E)$ such that $(q_0, 0) \in V$ and for every $(q, i) \in V$, $\mathit{sucs}(q, i) \models \delta(q, u(i))$. 
That is, the run DAG starts in the initial state $q_0$ at depth $0$, and for each node in the DAG the successors of each node satisfy the transition formula given by the transition function $\delta$.

For example, consider a node $(q, i)$ and assume that $\delta(q, u(i)) = q_1 \lor q_2$.
Then the DAG must include either $q_1$ or $q_2$ (or both) as successors in the next level.
In particular, if the transition function $\delta$ uses only disjunctions (no conjunctions) the automataon is non-deterministic.
In this case, each node in the DAG can have a unique successor; the DAG is a line (a infinite sequence of states).
Conversely, if $\delta(q, u(i)) = q_1 \land q_2$ then both $q_1$ and $q_2$ need to appear in the next level, i.e., we need to construct accepting runs from both of these states.

A run  DAG $\mathbb{D}$ is accepting if for every infinite path in the DAG the minimal color that occurs infinitely many times (as given by $c$) is even.
We define $\calL(\calA) \subseteq \Sigma^\omega$ as all infinite words on which $\calA$ has an accepting run DAG. 

\section{Additional Material for Section \ref{sec:relationToLogics}}\label{app:relationToLogics}

We provide some background on the temporal logics \HyperLTL{}, \SL{}, \HyperATLS{}, and \HyperATLSS{} and give their full semantics.

\subsection{\SL{} and \HyperSL{}}\label{app:sl}

In this subsection, we provide addition details on the relation of \SL{} and \HyperSL{} (cf.~\Cref{sec:translation_sl_into_hypersl}).
We consider a variant of \SL{} that allows multiple agents to share the strategy by using explicit agent binding, but disallows agent bindings under temporal operators.

\begin{remark}
	The strategy logic by \citet{MogaveroMPV14} allows arbitrary nesting of agent binding within temporal operators.
	Our variant is equivalent to \SL{}[BG] \cite{MogaveroMPV14} -- a fragment that follows a strict separation between state and path formulas and thereby forbids agent binding under temporal operators. 
	Note that \SL[BG] strictly subsumes the strategy logic by Chatterjee et al.~\cite{ChatterjeeHP10}.
	We choose to restrict to the \SL{}[BG] fragment as it is syntactically closer to temporal logics like \CTLS{} and \HyperCTLS{} and, thereby, also to \HyperSL{}.
\end{remark}

State and path formulas in \SL{} are defined as follows:
\begin{align*}
	\psiSL &:= a \mid \varphiSL \mid \neg \psiSL \mid \psiSL \land \psiSL \mid \ltlN \psiSL \mid \psiSL \ltlU \psiSL  \\
	\varphiSL &:= \psiSL \mid \varphiSL \land \varphiSL \mid \varphiSL \lor \varphiSL \mid \forall x\ldot \varphiSL \mid \exists x\ldot \varphiSL \mid (\agent{i}, x) \varphiSL 
\end{align*}
where $a \in \ap$, $x \in \stratVars$, and $\agent{i} \in \agents$. 
Importantly, we assume that every state formula occurring in a path formula is closed.

The idea of  \SL{} is to separate the quantification of a strategy and the binding of a strategy to some agent. 
To accomplish the latter, it features an explicit agent-binding construct $(\agent{i}, x) \varphiSL$ which evaluates $\varphiSL$ after binding agent $\agent{i}$ to strategy $x$. 

\paragraph{Semantics}
Assume $\calG = (S, s_0, \moves, \kappa, L)$ is a fixed CGS.
Given a path $p \in S^\omega$ we define the semantics of path formulas as expected:
\begin{align*}
	p &\models_\calG a &\text{iff } \quad&a \in L(p(0))\\
	p &\models_\calG \varphiSL &\text{iff } \quad&p(0), \{\}, \{\} \models_\calG\varphiSL\\
	p&\models_\calG \psiSL_1 \land \psiSL_2 &\text{iff } \quad &p \models_\calG \psiSL_1 \text{ and } p \models_\calG \psiSL_2\\
	p &\models_\calG \neg \psiSL &\text{iff } \quad &p \not\models_\calG \psiSL\\
	p &\models_\calG \ltlN \psiSL &\text{iff } \quad &p[1, \infty] \models_\calG \psiSL\\
	p &\models_\calG \psiSL_1 \ltlU \psiSL_2 &\text{iff } \quad &\exists j \in \nat \ldot p[j, \infty] \models_\calG \psiSL_2 \text{ and } \\
	&  \quad\quad\quad\quad\quad\quad\quad\quad\quad\quad \forall 0 \leq k < j\ldot p[k, \infty] \models_\calG \psiSL_1 \span \span
\end{align*}
In the semantics of state formulas, we keep track of a strategy for each strategy variable via a (partial) strategy assignment $\Delta : \stratVars \rightharpoonup \strats{\calG}$, similar to the \HyperSL{} semantics.
As \SL{} works with explicit agent bindings, we also keep track of a strategy for each agent using a (partial) function $\Theta : \agents \rightharpoonup \strats{\calG}$. 
We can then define:
\begin{align*}
	&s, \Delta, \Theta \models_\calG \forall x \ldot\varphiSL &\text{iff} \quad &\forall f \in\strats{\calG} \ldot s, \Delta[x \mapsto f], \Theta \models_\calG \varphiSL\\
	&s, \Delta, \Theta \models_\calG \exists x \ldot\varphiSL &\text{iff} \quad &\exists f \in\strats{\calG} \ldot s, \Delta[x \mapsto f], \Theta \models_\calG \varphiSL\\
	&s, \Delta, \Theta \models_\calG (\agent{i}, x)\varphiSL &\text{iff} \quad &s, \Delta, \Theta[\agent{i} \mapsto \Delta(x)] \models_\calG \varphiSL\\
	&s, \Delta, \Theta \models_\calG \varphiSL_1 \land  \varphiSL_2 &\text{iff} \quad &s, \Delta, \Theta \models_\calG \varphiSL_1 \text{ and } s, \Delta, \Theta \models_\calG \varphiSL_2\\
	&s, \Delta, \Theta \models_\calG \varphiSL_1 \lor  \varphiSL_2 &\text{iff} \quad &s, \Delta, \Theta \models_\calG \varphiSL_1 \text{ or } s, \Delta, \Theta \models_\calG \varphiSL_2\\
	&s, \Delta, \Theta \models_\calG \psiSL &\text{iff} \quad &\play_\calG\big(s, \prod_{\agent{i} \in \agents} \Theta (\agent{i}) \big) \models_\calG \psiSL
\end{align*}
Strategy quantification updates the binding in $\Delta$ (as in \HyperSL{}), whereas strategy binding updates the assignment of agents in $\Theta$. 
For each path formula we use the strategy profile $\Theta$ (mapping each agent to a strategy) to construct the path on which we evaluate $\psi$.
We write $\calG \models_{\text{\SL}} \varphiSL$ if $s_0, \{\}, \{\} \models_\calG \varphiSL$ in the \SL{} semantics.

\paragraph{Encoding}

We can easily encode \SL{} into \HyperSL{}.
Instead of using explicit agent bindings as in \SL{}, in \HyperSL{}, each path is annotated with an explicit strategy profile that assigns a strategy (variable) to each agent.
We assume that in the \SL{} formula no two quantifiers quantify the same strategy variable (which we can always ensure by $\alpha$-renaming).
Let $\dot{\pi}$ denote some \emph{fixed} path variable.
In our translation, we maintain an auxiliary mapping $\bX : \agents \rightharpoonup \stratVars$ mapping agents to strategy variables.
We then translate path formulas as follows:
\begin{align*}
	\slToHyper{a} &:= a_{\dot{\pi}} &  \slToHyper{\neg \psiSL} &:= \neg \slToHyper{\psiSL} & \slToHyper{\psiSL_1 \land \psiSL_2} &:= \slToHyper{\psiSL_1} \land \slToHyper{\psiSL_2}\\
	\slToHyper{\varphiSL} &:=  \slToHyper{\varphiSL}^{\{\}}  & \slToHyper{\ltlN \psiSL} &:= \ltlN \slToHyper{\psiSL} &  \slToHyper{\psiSL_1 \ltlU \psiSL_2} &:= \slToHyper{\psiSL_1} \ltlU \slToHyper{\psiSL_2}
\end{align*}
For state formulas we define:
\begin{align*}
	\slToHyper{\varphiSL_1 \land \varphiSL_2}^\bX &:= \slToHyper{\varphiSL_1}^\bX \land \slToHyper{\varphiSL_2}^\bX &\slToHyper{\varphiSL_1 \lor \varphiSL_2}^\bX &:= \slToHyper{\varphiSL_1}^\bX \lor \slToHyper{\varphiSL_2}^\bX\\
	 \slToHyper{\forall x\ldot \varphiSL}^\bX &:= \forall x\ldot \slToHyper{\varphiSL}^\bX  & \slToHyper{\exists x\ldot \varphiSL}^\bX &:= \exists x\ldot \slToHyper{\varphiSL}^\bX \\
	\slToHyper{\psiSL}^\bX &:= \slToHyper{\psiSL}[\dot{\pi} : \bX] & \slToHyper{(\agent{i}, x) \varphiSL}^\bX &:= \slToHyper{\varphiSL}^{\bX[\agent{i} \mapsto x]}	
\end{align*}
For path formulas, we resolve all atomic propositions on the fixed path variable $\dot{\pi}$. 
For state formulas, we record the strategy variable played by each agent in the function $\bX$, and whenever we reach a path formula, use $\bX$ to construct the fixed path $\dot{\pi}$. 
Note that the formula resulting from the translation uses boolean combination of state formulas. 
As already argued in \Cref{sec:hypersl}, we can bring such formulas into the \HyperSL{} syntax by introducing multiple paths (see \Cref{ex:state-conjunctions}). 
We can show that our translation is correct and thus prove \Cref{lem:sl-in-hypersl}:

\reSlToHyper*
\begin{proof}
	We can easily show that for any game structure $\calG$ and \SL{} formula $\varphiSL$, we have that $\calG \models_{\text{\SL}} \varphiSL$ if an only if $\calG \models \slToHyper{\varphiSL{}}^{\{\}}$. 
\end{proof}

\subsection{\HyperATLS{} and \HyperSL{}}\label{app:hyperatl}

In this subsection, we provide addition details on the relation of \HyperATLS{} and \HyperSL{} (cf.~\Cref{sec:translation_hyperatl_into_hypersl}).

Our hyper variant of strategy logic considers strategies as first-order objects that can be compared in different strategy profiles. 
A weaker form of strategic reasoning is offered in \ATLS{} \cite{AlurHK02} by only reasoning about the \emph{outcome} of a strategic interaction. 
The \ATLS{} formula $\llangle A \rrangle \psi$ expresses that the agents in $A$ have a joint strategy to enforce that the system follows some path that satisfies $\psi$.
\HyperATLS{} \cite{BeutnerF21,BeutnerF23} is a hyper-variant of ATL that combines strategic reasoning with the ability to express hyperproperties. 
Similar to \ATLS{}, \HyperATLS{} only considers the outcomes of a strategy but binds this outcome to a path variable which allows comparison w.r.t. a hyperproperty. 
Formulas in \HyperATLS{} are generated by the following grammar:
\begin{align*}
	\psiATL &:= a_\pi \mid \neg \psiATL \mid \psiATL \land \psiATL \mid \ltlN \psiATL \mid \psiATL \ltlU \psiATL \\
	\varphiATL &:=\llangle A \rrangle \pi\ldot \varphiATL \mid \llbracket A \rrbracket \pi\ldot \varphiATL \mid \psiATL
\end{align*}
where $a \in \ap$, $\pi \in \pathVars$, and $A \subseteq \agents$.

\paragraph{Semantics}

The semantics of \HyperATLS{} operates on a path assignment $\Pi : \pathVars \rightharpoonup S^\omega$.
For path formulas, we follow a similar semantics as used in \HyperSL{} (cf.~\Cref{sec:hypersl}):
\begin{align*}
	\Pi &\models_\calG a_\pi &\text{iff } \quad&a \in L(\Pi(\pi)(0))\\
	\Pi &\models_\calG \psiATL_1 \land \psiATL_2 &\text{iff } \quad &\Pi \models_\calG \psiATL_1 \text{ and } \Pi \models_\calG \psiATL_2\\
	\Pi &\models_\calG \neg \psiATL &\text{iff } \quad &\Pi \not\models_\calG \psiATL\\
	\Pi &\models_\calG \ltlN \psiATL &\text{iff } \quad &\Pi[1, \infty] \models_\calG \psiATL\\
	\Pi &\models_\calG \psiATL_1 \ltlU \psiATL_2 &\text{iff } \quad &\exists j \in \nat\ldot \Pi[j, \infty] \models_\calG \psiATL_2 \text{ and } \\
	&  \quad\quad\quad\quad\quad\quad\quad\quad\quad\quad \forall 0 \leq k < j\ldot \Pi[k, \infty] \models_\calG \psiATL_1 \span \span
\end{align*}
For state formulas, we need to consider all possible outcomes under strategies for a subset of the agents. 
Given $A \subseteq \agents$ and strategies $\{f_\agent{i} \mid \agent{i} \in A\}$ we define $\mathit{out}_\calG\big(s,\{f_\agent{i} \mid \agent{i} \in A\}\big) \subseteq S^\omega$ as
\begin{align*}
	\mathit{out}_\calG\big(s,&\{f_\agent{i} \mid \agent{i} \in A\}\big) := \\
	&\Big\{\play_\calG(s, \prod_{\agent{i} \in \agents} f_\agent{i})\mid \forall \agent{i} \in \agents \setminus A\ldot  f_\agent{i} \in \strats{\calG} \Big\}.
\end{align*}
That is, $\mathit{out}_\calG\big(s,\{f_\agent{i} \mid \agent{i} \in A\}\big)$ contains all possible paths compatible with the strategies in $\{f_\agent{i} \mid \agent{i} \in A\}$.
We can then evaluate a state formula in the context of a state $s$ and path assignment $\Pi$.
\begin{align*}
	s, \Pi &\models_\calG \psiATL \quad &&\text{iff} \quad \Pi \models_\calG \psiATL\\
	s, \Pi &\models_\calG \llangle A \rrangle \pi\ldot \varphiATL &&\text{iff} \quad \exists \{f_\agent{i} \mid \agent{i} \in A\}\ldot\\
	& \quad\quad \quad \forall p \in \mathit{out}_\calG\big(s, \{f_\agent{i} \mid \agent{i} \in A\}\big) \ldot s, \Pi[\pi \mapsto p] \models_\calG \varphiATL \span\span\\
	s, \Pi &\models_\calG \llbracket A \rrbracket \pi\ldot \varphiATL &&\text{iff} \quad \forall \{f_\agent{i} \mid \agent{i} \in A\}\ldot\\
	& \quad\quad \quad \exists p \in \mathit{out}_\calG\big(s, \{f_\agent{i} \mid \agent{i} \in A\}\big) \ldot s, \Pi[\pi \mapsto p] \models_\calG \varphiATL \span\span
\end{align*}
That is, a formula $\llangle A \rrangle \pi\ldot \varphiATL$ holds when there exist strategies for all agents in $A$ such that all possible outcomes under those fixed strategies, when bound to $\pi$, satisfy $\varphiATL$.
Likewise, $\llbracket A \rrbracket \pi\ldot \varphiATL$ holds when all possible strategies admits some path that, when bound to $\pi$, satisfies $\varphiATL$.
We write $\calG \models_{\text{\HyperATLS}} \varphiATL$ if $s_0,\{\} \models_\calG \varphiATL$ in the \HyperATLS{} semantics. 

\paragraph{Encoding}

Similar to the fact that \ATLS{} can be encoded in \SL{} \cite{ChatterjeeHP10}, we can encode \HyperATLS{} in \HyperSL{}.
The idea is to treat the ATL-quantifier $\llangle A \rrangle$ as a strategy quantifier that existentially quantifies over strategies for agents in $A$ and then universally over strategies for agents outside of $A$.
\HyperATLS{} path formulas can be translated verbatim into \HyperSL{} path formulas.
To translate state formulas we maintain a auxiliary mapping $\bL : \pathVars \rightharpoonup \stratVars^\agents$ from path variables to strategy profiles:
\begin{align*}
	\AtlToHyper{\psiATL}^\bL &:= \psiATL [\pi_1 : \bL(\pi_1), \ldots, \pi_m : \bL(\pi_m)]\\
	\AtlToHyper{\llangle A \rrangle \pi. \varphiATL}^\bL &:= \bigexists_{\agent{i} \in A} x_\agent{i}\ldot \bigforall_{\agent{i} \in \agents \setminus A}\!\!\! x_\agent{i}\ldot \, \AtlToHyper{\varphiATL}^{\bL[\pi \mapsto \prod_{\agent{i} \in \agents} x_\agent{i}  ]}\\
	\AtlToHyper{\llbracket A \rrbracket \pi. \varphiATL}^\bL &:= \bigforall_{\agent{i} \in A} x_\agent{i}\ldot \bigexists_{\agent{i} \in \agents \setminus A}\!\!\! x_\agent{i}\ldot \, \AtlToHyper{\varphiATL}^{\bL[\pi \mapsto \prod_{\agent{i} \in \agents} x_\agent{i}  ]}
\end{align*}
In the first case, we assume that $\pi_1, \ldots, \pi_m$ are the path variables used in $\psiATL$. 
In the second and third case, we assume that $x_\agent{1}, \ldots, x_\agent{n}$ are fresh strategy variables for each agent. 
Intuitively, we replace each $\llangle A \rrangle \pi$ quantifier with existential quantification over fresh strategies for $A$, followed by universal quantification over strategies for $\agents \setminus A$. 
In the auxiliary mapping $\bL$, we record which strategy profile we later want to use to construct $\pi$.
For the path formula $\psi$ we then reconstruct all paths used in the formula using the strategy profiles recorded in $\bL$. 
We can use our translation to prove \Cref{lem:hyperatl-in-hypersl}:

\hyperatlsToHyper*
\begin{proof}
	An easy induction shows that for any game structure $\calG$ and \HyperATLS{} formula $\varphiATL$ it holds that $\calG \models_{\text{\HyperATLS}} \varphiATL$ if and only if $\calG \models \AtlToHyper{\varphiATL}^{\{\}}$. 
\end{proof}

\paragraph{\HyperATLSS{}}

As our translation $\AtlToHyper{\varphiATL}^{\{\}}$ explicitly quantifies over strategies for all agents, we can easily enforce that two agents \emph{share} a strategy by simply using the same strategy (variable) for both.
Using a slight modification of the previous translation, we can thus handle the sharing constraints from \HyperATLSS{} \cite{BeutnerF24}:

\hyperatlssToHyper*

\subsection{\SLI{} and \HyperSL{}}\label{app:hypersl-and-sliii}

In this subsection, we provide addition details on the relation of \SLI{} and \HyperSL{} (cf.~\Cref{sec:hyperslAndSLI}).

\SLI{} \cite{BerthonMMRV17} extends \SL{} (cf.~\Cref{sec:translation_sl_into_hypersl}) by allowing strategies that only observe parts of the system. 
The formal model of \SLI{} are game structures that are endowed with an observations.  
Let $\calG = (S, s_0, \moves, \kappa, L)$ be a fixed game structure and let $\obs$ be a fixed finite set of so-called \emph{observations}.
An \emph{observation family} $\{\sim_o\}_{o \in \obs}$ associates an equivalence relation $\sim_o \subseteq S \times S$ with each $o \in \obs$.

For a strategy with observation $o$, two states $s \sim_o s'$ appear identical. 
This naturally extends to finite plays:
Two finite plays $p, p' \in S^+$ are $o$-indistinguishable, written $p \sim_o p'$, if $|p| = |p'|$ and for each $0 \leq i < |p|$, $p(i) \sim_o p'(i)$.
An \emph{$o$-strategy} is a function $f : S^+ \to \moves$ that cannot distinguish between $o$-indistinguishable plays, i.e., for all $p, p' \in S^+$ with $p \sim_o p'$ we have $f(p) = f(p')$. 
We denote with $\strats{\calG, o}$ the set of all $o$-strategies in $\calG$. 
We consider \SLI{} formulas that are generated by the following grammar:
\begin{align*}
	\psiSLI &:= a \mid \neg \psiSLI \mid \psiSLI \land \psiSLI \mid \ltlN \psiSLI \mid \psiSLI \ltlU \psiSLI \\
	\varphiSLI &:= \psiSLI \mid \varphiSLI \land \varphiSLI \mid \varphiSLI \lor \varphiSLI \mid \forall x^o\ldot \varphiSLI \mid \exists x^o\ldot \varphiSLI \mid (\agent{i}, x) \varphiSLI
\end{align*}
where $a \in \ap$, $x \in \stratVars$, $\agent{i} \in \agents$, and $o \in \obs$ is an observation.
Compared to \SL{}, strategy quantification in \SLI{} does not range over arbitrary strategies but over strategies that respect a given observation. 

\paragraph{Semantics}
The semantics of \SL{} (cf.~\Cref{app:sl}) only needs to be changed in the case of strategy quantification:
\begin{align*}
	&s, \Delta, \Theta \models_\calG \forall x ^o\ldot\varphiSLI  &\text{iff} \quad &\forall f \in\strats{\calG, o} \ldot s, \Delta[x \mapsto f], \Theta \models_\calG \varphiSLI\\
	&s, \Delta, \Theta \models_\calG \exists x ^o\ldot\varphiSLI  &\text{iff} \quad &\exists f \in\strats{\calG, o} \ldot s, \Delta[x \mapsto f], \Theta \models_\calG \varphiSLI
\end{align*}
where we restrict quantification to $o$-strategies. 
Given a game structure $\calG$, a family $\{\sim_o\}_{o \in \obs}$, and an \SLI{} formula $\varphiSLI$ we write $(\calG, \{\sim_o\}_{o \in \obs}) \models_{\text{\SLI}} \varphiSLI$ if $s_0, \{\}, \{\} \models_\calG \varphiSLI$ in the \SLI{} semantics. 
See \cite{BerthonMMRV17,BerthonMMRV21} for concrete examples of \SLI{}.

\paragraph{Injective Labeling and Action Recording}

To prove \Cref{prop:sliToHyper}, we need to translate \SLI{} MC instances into equisatisfiable \HyperSL{} instances.
In order to translate the model-checking instances, we first modify the underlying game structure.
The reason for this is simple:
Strategies are defined as functions $S^+ \to \moves$ and $\sim_o$ is defined as a direct relation on states, i.e., both are defined directly on \emph{components of the game structure}.
In contrast, \emph{within} our logic, we only observe the evaluation of the atomic propositions.
In a first step, we thus modify the game structure provided us with sufficient information within its atomic propositions:

\begin{definition}\label{def:IL_and_MR}
	A game structure $\calG = (S, s_0, \moves, \kappa, L)$ is \emph{injectively labeled} (IL), if $L : S \to 2^\ap$ is injective, i.e., two labels are equal iff the state is equal.
	A game structure is \emph{action recording} (AR) if for each agent $\agent{i} \in \agents$ and every action $a \in \moves$,  there exists an atomic proposition $\langle \agent{i}, a \rangle \in \ap$ that holds in a state exactly when $\agent{i}$ played action $a$ in the last step.
	That is, for all $s \in S$ and all action profile $\prod_{\agent{i} \in \agents} a_\agent{i}$,  we have
	\begin{align*}
		L\big(\kappa(s, \prod_{\agent{i} \in \agents} a_\agent{i})\big) \cap \{ \langle \agent{i}, a \rangle \mid \agent{i} \in \agents, a \in \moves \} = \bigcup_{\agent{i} \in \agents} \{ \langle \agent{i}, a_\agent{i}  \rangle \}.
	\end{align*}
\end{definition}

We can always ensure that $\calG$ is injectively labeled by adding at most $\lceil \log |S|\rceil$ fresh atomic propositions.
The evaluation of all formulas (which do not refer to these fresh propositions) remains unchanged when adding such fresh propositions.
Similarly, we can always ensure that a structure records actions.
This increases the number of states by a factor of $|\moves|^{|\agents|}$ but does not increase the branching in the system, nor does it change the semantics of \SLI{} if we extend $\sim_o$ to the new states in the obvious way:

\begin{lemma}\label{lem:to_move_recording_gsif}
	Given an \SLI{} MC instance $(\calG, \{\sim_o\}_{o \in \obs}, \varphiSLI)$ there exists an effectively computable \SLI{} instance $(\calG', \{\sim'_o\}_{o \in \obs}, \varphiSLI')$ where \textbf{(1)} $\calG'$ is IL and AR, and \textbf{(2)} $(\calG, \{\sim_o\}_{o \in \obs}) \models_{\text{\SLI{}}} \varphiSLI$ iff $(\calG', \{\sim'_o\}_{o \in \obs}) \models_{\text{\SLI{}}} \varphiSLI'$.
\end{lemma}
\begin{proof}
	Assume $\calG = (S, s_0, \moves, \kappa, L)$.
	Define $\ap' := \ap \uplus \{\langle \agent{i}, a \rangle  \mid \agent{i} \in \agents, a \in \moves\}$.
	We then define
	\begin{align*}
		\textstyle\calG' = (S \times (\agents \to \moves), (s_0, \prod_{\agent{i} \in \agents} a), \moves, \kappa', L')
	\end{align*}
	where $a \in \moves$ is some arbitrary action (in the initial state, we do not need to track the last played action). 
	For an action profile $\prod_{\agent{i} \in \agents} a_\agent{i}$, we define $\kappa'$ and $l'$ by
	\begin{align*}
		\kappa' \big((s, \_), \prod_{\agent{i} \in \agents} a_\agent{i}\big) &:= (\kappa(s, \prod_{\agent{i} \in \agents} a_\agent{i}), \prod_{\agent{i} \in \agents} a_\agent{i})\\
		L'(s, \prod_{\agent{i} \in \agents} a_\agent{i}) &:= L(s) \uplus \{ \langle \agent{i}, a_\agent{i} \rangle \mid \agent{i} \in \agents\}.
	\end{align*}
	The idea behind $\calG'$ is that we record the action profile that was last used in the second component of each state. 
	In each transition, we ignore the action profile in the current step, and record the new action profile n the second component. 
	In the labeling function, we can then use the action profile $\prod_{\agent{i} \in \agents} a_\agent{i}$ in the second component to set the APs $\langle \agent{i}, a_\agent{i} \rangle$ for all $\agent{i} \in \agents$.
	We define $\{\sim'_o\}_{o \in \obs}$ by 
	\begin{align*}
		\sim'_o := \Big\{ \big((s, \_), (s', \_)\big) \mid s \sim_o s' \Big\}.
	\end{align*}
	That is, for any observation cannot distinguish states based on the second position. 
	In particular note that $(s, \_)$ and $(s, \_)$ are always indistinguishable, i.e., all states we expanded that we added are indistinguishable under the new observation. 
	
	It is easy to see that $\calG'$ is AR and that
	$(\calG, \{\sim_o\}_{o \in \obs}) \models_{\text{\SLI{}}} \varphiSLI$ iff $(\calG', \{\sim'_o\}_{o \in \obs}) \models_{\text{\SLI{}}} \varphiSLI$.
	Note that we did not change the formula.
	In a second step, we can ensure that $\calG'$ is also IL by simply adding sufficiently many new propositions. 
	As those new propositions are never used in $\varphiSL$ the semantics is unchanged.
\end{proof}

\paragraph{Identifying Indistinguishable States}

In the next step, we construct a formula that identifies pairs of states that are indistinguishable according to $\sim_o$.
For $o \in \obs$, we define formula $\mathit{ind}_o(\pi_1, \pi_2)$ as follows:
\begin{align*}
	\!\!\bigvee_{(s, s') \in \sim_o} \Big(\bigwedge_{a \in L(s)} a_\pi \land \!\!\!\!\! \bigwedge_{a \in \ap\setminus L(s)} \!\!\!\!\! \neg a_\pi \land \!\!\!\!\! \bigwedge_{a \in L(s')} a_{\pi'} \land \!\!\!\!\! \bigwedge_{a \in \ap\setminus L(s')} \!\!\!\!\! \neg a_{\pi'}\Big)
\end{align*}
It is easy to see that on any injectively labeled game structure $\mathit{ind}_o(\pi_1, \pi_2)$ holds if the two paths bound to $\pi_1, \pi_2$ are $\sim_o$-related in their first state.

\paragraph{Enforcing Partial Information}

Given an observation $o \in \obs$, and strategy variable $x$, we define a formula $\mathit{indStrat}_o(x)$ that holds on a strategy iff this strategy is an $o$-strategy (in all reachable situations and for all agents) as follows:
\begin{align*}
	\mathit{indStrat}_o(x) &:= \forall y_1, \ldots, y_n, y'_1, \ldots, y'_n. \span\span \\
	&\bigwedge_{\agent{i} = 1}^n \psiSLI^{\agent{i}}_o \left[\begin{matrix}
		\pi_1 : (y_1, \ldots, y_{i-1}, x, y_{i+1}, \ldots, y_n )\\
		\pi_2 : (y'_1, \ldots, y'_{i-1}, x, y'_{i+1}, \ldots, y'_n )
	\end{matrix}\right]
\end{align*}
where
\begin{align*}
	\psiSLI_o^{\agent{i}} := \Big(\ltlN \Big(\bigwedge_{a \in \moves} &\langle  \agent{i}, a\rangle_{\pi_1} \leftrightarrow \langle \agent{i}, a \rangle_{\pi_2} \Big) \Big) \ltlW \Big( \neg \mathit{ind}_o(\pi_1, \pi_2) \Big).
\end{align*}
The path formula $\psiSLI_o^{\agent{i}}$ compares two paths $\pi_1, \pi_2$ and states that as long as a prefix on those to paths is $o$-indistinguishable (i.e., $\mathit{ind}_o(\pi_1, \pi_2)$ holds in each step), the action selected by agent $\agent{i}$ is the same on both paths (using the fact that the structure records actions). 
As we do not know which agents might end up playing strategy $x$ we assert that $x$ behaves as a $o$-strategy for all agents.
For each $\agent{i} \in \agents$ we thus compare two paths where $\agent{i}$ plays $x$, but all other agents play some arbitrary strategy, and assert that $\psiSLI_o^{\agent{i}}$ holds for those two paths. 
Strategy $x$ must thus respond to two $o$-indistinguishable prefixes with the same action in all reachable situations for all agents.

\paragraph{The Translation}

Using $\mathit{indStrat}_o(x)$ as a building block, we can modify the translation of \SL{} into \HyperSL{} from \Cref{app:sl}, and instead translate the much stronger \SLI{}.

\sliToHyper*
\begin{proof}
	In a first step, we use \Cref{lem:to_move_recording_gsif} to ensure that $\calG$ is IL and AR. 
	We can then translate $\varphi$ using a similar translation to the one used in \Cref{app:sl}.
	The only cases that require changening are the translation of quantification:
	\begin{align*}
		\slIToHyper{\forall x^o. \varphiSLI}^\bX &:= \forall x\ldot \mathit{indStrat}_o(x) \rightarrow \slIToHyper{\varphiSLI}^\bX \\
		\slIToHyper{\exists x^o. \varphiSLI}^\bX &:= \exists x \ldot  \mathit{indStrat}_o(x)  \land \slIToHyper{\varphiSLI}^\bX
	\end{align*}
	Note that the resulting formula uses implications between state formulas which is not supported by the HyperSL syntax. It is, however, easy to see that we can push the boolean operations into the path formula as observed in \Cref{ex:state-conjunctions}.
	
	We claim that for any  \SLI{} MC instance $(\calG,\{\sim_o\}_{o \in \obs}, \varphiSLI)$ where $\calG$ is IL and AR, we have that $(\calG,\{\sim_o\}_{o \in \obs}) \models_{\text{\SLI{}}} \varphiSLI$ iff $\calG \models \slIToHyper{\varphiSLI}^{\{\}}$.
	
	To prove the above we need to argue that $\mathit{indStrat}_o(x)$ really expresses that $x$ is a $o$-strategy. 
	It is easy to see that for any strategy $f \in \strats{\calG}$ we have $s, [x \mapsto f] \models \mathit{indStrat}_o$ if and only if $f$ is a $o$-strategy in all reachable situations from $s$. 
	That is, $f$ does not necessarily behave as an $o$-strategy in all situations, but at least in those situations that are actually reachable under $f$. 
	As any strategy will only ever be queried on plays that are compatible with the strategy itself, this suffices to encode the \SLI{} semantics. 
\end{proof}

\section{Additional Material for Section \ref{sec:decideable}}\label{app:decMaterial}

In this section we  prove the correctness of our \HyperSLP{} model-checking algorithm (\Cref{alg:main-alg}).
Our algorithm hinges on \lstinline[style=default, language=custom-lang]|simulate| procedure (\Cref{alg:product}) and the resulting properties (\Cref{prop:sim}).
We dedicate the entire \Cref{app:simProof} to a proof of \Cref{prop:sim}, and here focus on the correctness (and complexity) of \Cref{alg:main-alg}.

\subsection{Correctness Proof of \Cref{alg:main-alg}}

As already argued in the main part of the paper, our correctness proof relies on an inductive argument that establishes that we compute $(\calG, \dot{s}, k)$-summaries for each $1 \leq k \leq m + 1$. 

\lemmaInd*
\begin{proof}
	We show the statement by induction on $1 \leq k \leq m + 1$ (from $k=m+1$ to $k = 1$).
	For the base case ($k = m + 1$), we observe that the APA $\calA_{m+1}$ (computed in line \ref{line:ltlToApa}) is a $(\calG, \dot{s}, m+1)$-summary.

	For the induction step we can assume -- by induction hypothesis -- that prior to line \ref{line:product}, $\calA_{k+1}$ is a $(\calG, \dot{s}, k+1)$-summary. 
	Recall that $\calA_{k+1}$ is an APA over $(\{\pi_1, \ldots \pi_{k}\}  \to S)$ and $\calA_k$ over $(\{\pi_1, \ldots \pi_{k-1}\}  \to S)$.	
	We claim that $\calA_k$ is a $(\calG, \dot{s}, k)$-summary.
	To show this, take any $\Pi : \{\pi_1, \ldots \pi_{k-1}\} \to S^\omega$ and we need to show that (cf.~\Cref{def:equiv} of $(\calG, \dot{s}, k)$-summary) $\mathit{zip}(\Pi) \in \calL(\calA_k)$ if and only if 
	\begin{align*}
		&\widetilde{\flat_k}\cdots \widetilde{\flat_m}\ldot \Pi\Big[\pi_j\mapsto \play_\calG\big(\dot{s}, \prod_{\agent{i} \in \agents} f_{\vec{x}_j(\agent{i})}\big)\Big]_{j = k}^m \models_\calG \psi.
	\end{align*}
	By adding parenthesis, the latter holds if and only if 
			\begin{align*}
				&\widetilde{\flat_k}. \Big(\widetilde{\flat_{k+1}}\cdots \widetilde{\flat_m}\ldot \Pi\Big[\pi_j\mapsto \play_\calG\big(\dot{s}, \prod_{\agent{i} \in \agents} f_{\vec{x}_j(\agent{i})}\big)\Big]_{j = k}^m \models_\calG \psi \Big)
			\end{align*}
	which holds if and only if 
	\begin{align*}
		&\widetilde{\flat_k}. \Big(\widetilde{\flat_{k+1}}\cdots \widetilde{\flat_m}\ldot \Pi'\Big[\pi_j \mapsto \play_\calG\big(\dot{s}, \prod_{\agent{i} \in \agents} f_{\vec{x}_j(\agent{i})}\big)\Big]_{j = k+1}^m \models_\calG \psi \Big)
	\end{align*}
	where $\Pi' = \Pi[\pi_k \mapsto \play_\calG\big(\dot{s}, \prod_{\agent{i} \in \agents} f_{\vec{x}_k(\agent{i})}\big)\big]$.
	By the assumption that $\calA_{k+1}$ is a $(\calG, \dot{s}, k+1)$-summary we can replace the inner part and get that the above is equivalent to 
	\begin{align*}
		&\widetilde{\flat_k}. \mathit{zip}(\Pi') \in \calL(\calA_{k+1}).
	\end{align*}
	After unfolding the definition of $\Pi'$, this becomes 
	\begin{align*}
		&\widetilde{\flat_k}\ldot  \mathit{zip}\Big(\Pi\big[\pi_k \mapsto \play_\calG\big(\dot{s}, \prod_{\agent{i} \in \agents} f_{\vec{x}_k(\agent{i})}\big)\big] \Big)\in \calL(\calA_{k+1}).
	\end{align*}
	Now recall that we defined $\calA_{k} = $~\lstinline[style=default, language=custom-lang]|simulate($\calG$,$\dot{s}$,$\pi_k$,$\bX_k$,$\flat_k$,$\calA_{k+1}$)|.
	By \Cref{prop:sim} we now have that the above holds iff 
	\begin{align*}
		\mathit{zip}(\Pi) \in \calL(\calA_{k}).
	\end{align*}
	as required. 
\end{proof}

\lemmaInit*
\begin{proof}
	Note that the alphabet of $\calA$ is the singleton set $(\emptyset \to S)$.
	We thus get that $\calL(\calA) \neq \emptyset$ iff $\mathit{zip}(\{\}) \in \calL(\calA)$ where $\{\}$ is the unique path assignment $\emptyset \to S^\omega$ so  $\mathit{zip}(\{\})$ is the unique word over $(\emptyset \to S)$. 
	Now by \Cref{def:equiv} we have that $\mathit{zip}(\{\}) \in \calL(\calA)$ iff
	\begin{align*}
			&\widetilde{\flat_1}\cdots \widetilde{\flat_m}\ldot \{\}\Big[\pi_j \mapsto \play_\calG\big(\dot{s}, \prod_{\agent{i} \in \agents} f_{\bX_j(\agent{i})}\big)\Big]_{j = 1}^m \models_\calG \psi
		\end{align*}
	where we add all paths to the empty path assignment $\{\}$.
	As we add to the empty path assignment, the above is thus equivalent to 
	\begin{align*}
		&\widetilde{\flat_1}\cdots \widetilde{\flat_m}\ldot \Big[\pi_j \mapsto \play_\calG\big(\dot{s}, \prod_{\agent{i} \in \agents} f_{\bX_j(\agent{i})}\big)\Big]_{j = 1}^m \models_\calG \psi
	\end{align*}
	which exactly expresses $\dot{s}, \{\} \models_\calG \varphi$ in the \HyperSL{} semantics.
\end{proof}

\subsection{Model-Checking Complexity}

\compTheo*
\begin{proof}
	Each time we invoke \lstinline[style=default, language=custom-lang]|simulate| (\Cref{alg:product}) the automaton size increases by two exponents. 
	A formula with block-rank $m$ requires $m$ applications of \lstinline[style=default, language=custom-lang]|simulate| (cf.~\Cref{alg:main-alg}), so the final automaton $\calA_1$ has size that is $2m$-times exponential in the size of $\psi$ and $\calG$.
	Automaton $\calA_1$ operates on a singleton alphabet ($\emptyset \to S$), so we can decide its emptiness in polynomial time. 
	Model checking is thus in $2m$-\EXPTIME{} in the size of $\psi$ and $\calG$.
\end{proof}

\lowerBound*
\begin{proof}
	For \HyperATLS{} it is known that checking a formula with $m$ quantifiers is $(2m-1)$-\EXPSPACE{}-hard (in the size of the formula) \cite[Thm.~7.1 and 7.2]{BeutnerF23}.
	As the translation of a \HyperATLS{} formula with $m$ qunatifiers into a \HyperSLP{} formula is linear (cf.~\Cref{lem:hyperatl-in-hypersl} and \Cref{app:hyperatl}) and yields a formula with block-rank $m$, the lower bound follows.
\end{proof}

\subsection{Beyond \HyperSLP{}}\label{app:beyond}

As we argued in \Cref{sec:beyond_spe}, any \HyperSL{} formula where the prefix cannot be grouped into blocks as in \Cref{def:spe}, MC becomes in general undecidable.

\begin{lemma}\label{lem:undecBoundry}
	Model checking for a \HyperSL{} formula of the form 
	\begin{align*}
		\exists x. \exists y. \forall z. \forall w\ldot \psi\left[\begin{aligned}
			\pi_1 &: (x, z)\\
			\pi_2 &: (y, w)
		\end{aligned}\right]
	\end{align*}
	is, in general, undecidable.
\end{lemma}
\begin{proof}
	We encode the \HyperLTL{} realizability problem of a $\forall^2$ formula, which is known to be undecidable \cite{FinkbeinerHLST18}.
	Given a \HyperLTL{} formula $\varphiH = \forall \pi_1. \forall \pi_2. \psiH$ over $I \uplus O$ we define 
	\begin{align*}
		\psi' := \psiH \land \Big( \Big(\ltlN \big(\bigwedge_{a \in O} a_{\pi_1} \leftrightarrow a_{\pi_2} \big) \Big)\ltlW \big(\bigvee_{a \in I} a_{\pi_1} \not\leftrightarrow a_{\pi_2} \big)\Big)
	\end{align*}
	That is, the two paths $\pi_1, \pi_2$ should satisfy $\psiH$ and, in addition, when given the same sequence of inputs, the output should be the same. 
	We claim that $\varphiH = \forall \pi_1. \forall \pi_2. \psiH$ is realizable if and only if 
	\begin{align*}
		\calG_{(I, O)} \models \exists x. \exists y. \forall z. \forall w\ldot \psi'\left[\begin{aligned}
			\pi_1 &: (x, z)\\
			\pi_2 &: (y, w)
		\end{aligned}\right]
	\end{align*}
	where $\calG_{(I, O)}$ is the CGS in which agent $\agent{1}$ can set the inputs $I$ and agent $\agent{2}$ can set the outputs $O$ (see, e.g., \cite{AlurHK02}). 
	The intuition is that the additional conjunct requires that the two strategies bound to $x$ and $y$ denote the \emph{same} strategy.
	The above formula thus states that there exists some strategy that controls the outputs such that all pairs of traces under that strategy satisfy $\psiH$.
	Deciding the existence of such a strategy is undecidable \cite{FinkbeinerHLST18}, so we get the desired result.
\end{proof}

\section{Proof Of Proposition \ref{prop:sim}}\label{app:simProof}

In this section, we prove \Cref{prop:sim}:

\simProp*

For any $m \in \nat$, we define $[m] := \{1, \ldots, m\}$. 
To make the proof of its correctness easier, we assume that 
\begin{align}\label{eq:strat_prefix}
	\flat = \forall 1.\exists 2.\forall 3\ldots \exists {2m}.
\end{align}
That is, we assume that the strategy variables are number in $[2m]$.
And, moreover, we assume the quantifier block alternates strictly in every step: odd strategy variables (numbers) are universally quantified and even variables are existentially quantified. 
Note that this assumption is w.l.o.g., we can always add quantification over additional strategy variables that wil never be used for the construction of  $\pi$.
Having a fixed alternation makes formal reasoning and notation easier. 

\subsubsection*{The Candidate}

Let $\calA_\mathit{det} = (Q, q_0, \delta, c)$ be the DPA constructed from $\calA$ in line \ref{line:toDPA} of \Cref{alg:product} (using \Cref{prop:toDPA}).
Following the construction in \Cref{alg:product}, we get that $\calB$ (the result of \lstinline[style=default, language=custom-lang]|simulate($\calG$,$\dot{s}$,$\pi$,$\bX$,$\flat$,$\calA$)|) satisfies
\begin{align*}
	\calB = \big(Q \times S, (q_0, \dot{s}), \delta', c'\big),
\end{align*}
where $c'(q, s) := c(q)$, and for $\bT \in S^{V}$ , $\delta'\big( (q, s) , \bT\big)$ is defined as as 
\begin{align*}
	\bigwedge_{a_{1} \in \moves } \bigvee_{a_{2} \in \moves } \cdots \bigvee_{a_{{2m}} \in \moves} \Big(\delta\big(q, \bT[\pi \mapsto s] \big), \kappa\big(s, \prod_{\agent{i} \in \agents} a_{\bX(\agent{i})}  \big)\Big),
\end{align*}
With this construction fixed, it remains to argue that it accepts the desired language. 
Expressed as a lemma:

\begin{lemma}\label{lem:taregt_for_B}
	For any path assignment $\Pi : V \to S^\omega$ we have $\mathit{zip}(\Pi) \in \calL(\calB)$  if and only if
	\begin{align}\label{eq:lemmaTarget}
		&\widetilde{\flat}\ldot  \mathit{zip}\Big(\Pi\big[\pi \mapsto \play_\calG(\dot{s}, \prod_{\agent{i} \in \agents} f_{\bX(\agent{i})}\Big)\big] \Big)\in \calL(\calA).
	\end{align}
\end{lemma}

\subsection{Concurrent Parity Games}

The main prove idea in showing that $\calB$ accepts the language desired by \Cref{prop:sim}, goes via the determinacy of concurrent parity games. 

\begin{definition}[Concurrent Parity Game]
	A concurrent parity game (CPG) is a a tuple $\calP = (V, v_0, \moves, m, \mu, c)$ where $V$ is a (possibly infinite) set of vertices, $v_0 \in V$ is an initial vertex, $\moves$ is a finite set of actions, $m \in \nat$ gives the number of alternations (so $2m$ is the number of player), $\mu : V \times ([2m] \to \moves) \to V$ is a transition function, and $c : V \to C$ is a coloring for some finite $C \subseteq \nat$.
\end{definition}

We refer to the protagonist in CPGs as players to distinguish them from the agents in a game structure. 
Similar to the simplyfying assumption we made of $\flat$, we will later quantify universally over strategies for odd players and existentially for even players. 
A strategy in $\calP$ is a function $f : V^+ \to \moves$ and we write $\strats{\calP}$ for the set of all strategies in $\calP$. 
When given a strategy profile $\bF : [2m] \to \strats{\calP}$ assigning a strategy for each player, and a vertex $v \in V$, we define $\play_\calP(v, \bF) \in V^\omega$ as the unique play $p \in V^\omega$ such that $p(0) = v$, and for every $j \in \nat$, $p(j+1) = \mu\big(p(j), \prod_{l \in [2m]} \bF(l)(p[0,j])\big)$ (similar to the definition in CGSs, cf.~\Cref{sec:prelim}).
We say an infinite play $p \in V^\omega$ is \emph{even} if the minimal color that occurs infinity many times (as given by $c$) is even.
In this case we write $\mathit{even}(p)$.

\begin{definition}\label{def:cpgSem}
	The CPG $\calP$ is \emph{won by the existential team} if
	\begin{align*}
		\forall f_{1} \in \strats{\calP}.&\exists f_{2} \in \strats{\calP}.\forall f_{3} \in \strats{\calP} \ldots \exists f_{2m} \in \strats{\calP} \ldot \\
		&\mathit{even}\Big( \play_\calP(v_0, \prod_{l \in [2m]} f_l) \Big).
	\end{align*}
\end{definition}

That is we quantify over strategies in $\calP$ for all player $l \in [2m]$; universally for odd player and existentially for even player in alternating fashion.
The game is won (by the existential team) if the existential quantifier can ensure that the resulting play is even.

\subsection{Construction of $\calP_\Pi$}

Assume we are given a fixed strategy assignment $\Pi : V \to S^\omega$. 
We design an infinite-state CPG $\calP_\Pi$ as
\begin{align*}
	\calP_\Pi = (Q \times S \times \nat, (q_0, \dot{s}, 0), \moves, m, \mu, c')
\end{align*}
where $Q$ is the set of states in $\calA_\mathit{det}$, $c'(q, s, N) := c(q)$ and for each action profile $\bA : [2m] \to \moves$ we define $\mu\big( (q, s, N), \bA \big)$ as 
\begin{align*}
	\Big(\delta\big(q, \mathit{zip}(\Pi)(N)[\pi \mapsto s] \big), \kappa\big(s, \prod_{\agent{i} \in \agents} \bA(\bX(\agent{i})) \big), N + 1 \Big)
\end{align*}
In the following we abbreviate $V := Q \times S \times \nat$ for the vertices in $\calP_\Pi$, and define $v_0 := (q_0, \dot{s}, 0)$ as the initial vertex of $\calP_\Pi$.

Let us discuss the construction of $\calP_\Pi$.
Actions in $\calP_\Pi$ are the same as in $\calG$ (i.e., $\moves$).
For each quantified strategy variable $1, \ldots, 2m$ (cf.~\Cref{eq:strat_prefix}), we have a corresponding player in $\calP_\Pi$ (i.e., there are $m$ alternations, so the set of players is $[2m]$).
We operate on vertices $(q, s, N)$, where $q \in Q$ and $s \in S$ are similar to the construction of $\calB$.
In addition we track the current step $N$.
To update $q$ we invoke the transition function $\delta$ of $\calA_\mathit{det}$ on the current automaton state and letter $\mathit{zip}(\Pi)(N)[\pi \mapsto s]$.
Note that this corresponds to the input of $\calB$: $\calB$ reads the zipping of a strategy assignment  $\Pi$ as an input and thus the $N$th letter of the input is $\mathit{zip}(\Pi)(N)$. 
That is, we ``hardcode'' $\Pi$ into the game. 
For this, we maintain the current step via counter $N$ and ``pretend'' the input in the $N$th step was $\mathit{zip}(\Pi)(N)$.
Similar to $\calB$, we update the automaton state by passing $\mathit{zip}(\Pi)(N)[\pi \mapsto s]$ to $\calA_\mathit{det}$'s transition function.
To update the state of the simulation of $\calG$ we proceed as in $\calB$, i.e., given a function $\bA : [2m] \to \moves$ that fixes actions for all strategy variables (or, equivalently, players in $\calP_\Pi$), each agent $\agent{i} \in \agents$ will simply play the action assigned to strategy variable $\bX(\agent{i})$, i.e., $\bA(\bX(\agent{i}))$.

The idea of $\calP_\Pi$ is to serve as intermediate representation between the target language (Eq.~(\ref{eq:lemmaTarget})) and the definition of $\calB$. 
In the semantics of CPGs the quantification over strategies for the players occurs ``outside'', i.e., strategies are fixed globally (cf.~\Cref{def:cpgSem}). 
As players in $\calP_\Pi$ correspond exactly to the strategy variables used in $\flat$ ($[2m]$), this ``outer'' quantification mimics the quantification found in Eq.~(\ref{eq:lemmaTarget}).
On the other hand, the updates of automaton and system state in $\calP_\Pi$ are similar to the updates performed in $\calB$.

\subsection{The First Implication}

As a first step, we will show that $\calP_\Pi$ is won by the existential player if and only $\Pi$ satisfies Eq.~(\ref{eq:lemmaTarget}).
This step is easy: The quantification in $\calP_\Pi$ and Eq.~(\ref{eq:lemmaTarget}) is very similar (i.e., occurs ``outside'' the game).
We can, therefore, transfer strategies between $\calG$ and $\calP_\Pi$ as follows:

\begin{lemma}\label{lem:strat_translation}
	The following holds:
	\begin{itemize}[leftmargin=*]
		\item 	For any $\Delta : [2m] \to \strats{\calP_\Pi}$ there exists a $\widetilde{\Delta} : [2m] \to \strats{\calG}$ such that 
		\begin{align*}
			& \mathit{zip}\Big(\Pi\big[\pi \mapsto \play_\calG\big(\dot{s}, \prod_{\agent{i} \in \agents} \widetilde{\Delta}({\bX(\agent{i})})\big)\big] \Big) \in \calL(\calA).
		\end{align*}
		if and only if $\mathit{even}(\play_{\calP_\Pi}(v_0, \Delta))$.
		
		\item For any $\Delta : [2m] \to \strats{\calG}$ there exists a $\widetilde{\Delta} : [2m] \to \strats{\calP_\Pi}$ such that 
		\begin{align*}
			& \mathit{zip}\Big(\Pi\big[\pi \mapsto \play_\calG\big(\dot{s}, \prod_{\agent{i} \in \agents} \Delta({\bX(\agent{i})})\big)\big] \Big) \in \calL(\calA).
		\end{align*}
		if and only if $\mathit{even}(\play_{\calP_\Pi}(v_0, \widetilde{\Delta}))$.
	\end{itemize}
\end{lemma}
\begin{proof}
	We show both claims separately.
	
	\begin{itemize}[leftmargin=*]
		\item We translate strategies in $\calP_\Pi$ to strategies in $\calG$.
		Let $f : V^+ \to \moves$ be some strategy in $\calP_\Pi$.
		We define $\widetilde{f} : S^+ \to \moves$ as follows:
		Let $u\in S^+$ be given.
		Define $\bQ \in Q^+$ as follows:
		We define $\bQ(0) = q_0$ (the initial state of $\calA_\mathit{det}$).
		For $0 \leq i < |u|$ we then define $\bQ(i+1) := \delta(\bQ(i), \mathit{zip}(\Pi)(i)[\pi \mapsto u(i)])$. 	
		The sequence $\bQ$ thus gives the unique state sequence in $\calA_\mathit{det}$ when reading the first $|u|$ states of $\mathit{zip}(\Pi)$ extended with $u$.
		Now consider the finite play in $\calP_\Pi$ 
		\begin{align*}
			\tau = \big(\bQ(0), u(0), 0\big) \cdots \big(\bQ(|u|-1), u(|u|-1), |u|-1\big).
		\end{align*}
		The intuition is that $\tau$ corresponds to the unique path in $\calP_\Pi$ that, when projected onto the system state, gives $u$.
		We define $\widetilde{f}(u) := f(\tau)$. 
		
		For any strategy assignment $\Delta : [2m] \to \strats{\calP_\Pi}$ in $\calP_\Pi$ we define $\widetilde{\Delta} : [2m] \to \strats{\calG}$ as the assignment obtained by applying~$\widetilde{\cdot}$~point wise. 
		
		It is now easy to see that $\play_\calG\big(\dot{s}, \prod_{\agent{i} \in \agents} \widetilde{\Delta}({\bX(\agent{i})})\Big)) \in S^\omega$ equals $\play_{\calP_\Pi}(v_0, \Delta)$ (when projecting vertices in $Q \times S \times \nat$ on $S$). 
		By construction, the automaton component in each vertex of $\calP_\Pi$ simply simulates $\calA_\mathit{det}$ on the generated sequence of system state and thus accepts play iff the corresponding automaton sequence is accepting. 
		We thus get that the 
		\begin{align*}
			& \mathit{zip}\Big(\Pi\big[\pi \mapsto \play_\calG\big(\dot{s}, \prod_{\agent{i} \in \agents} \widetilde{\Delta}({\bX(\agent{i})})\big)\big] \Big) \in \calL(\calA) = \calL(\calA_\mathit{det})
		\end{align*}
		if and only if $\mathit{even}(\play_{\calP_\Pi}(v_0, \Delta))$, as required.
		
		\item 
		Let $f : S^+ \to \moves$ be a strategy in $\calG$. 
		We define $\widetilde{f} : V^+ \to \moves$ in $\calP_\Pi$ as follows:
		For any $u \in V^+$ let $\bS \in S^+$ be the projection on the system state in $u$.
		We define $\widetilde{f}(u) := f(\bS)$, i.e., only query $f$ on the sequences of system states.
		For any strategy assignment $\Delta : [2m] \to \strats{\calG}$ in $\calP_\Pi$ we define $\widetilde{\Delta} : [2m] \to \strats{\calP_\Pi}$ as the assignment obtained by applying~$\widetilde{\cdot}$~point wise. 
		As in the first case, it is easy to see that $\play_\calG\big(\dot{s}, \prod_{\agent{i} \in \agents} \Delta({\bX(\agent{i})})\Big)) \in S^\omega$ is equal to $\play_{\calP_\Pi}(v_0, \widetilde{\Delta})$ (when projecting vertices in $Q \times S \times \nat$ on $S$). 
		As $\calP_\Pi$ simulates $\calA_\mathit{det}$ on the sequence of system states we thus get
		\begin{align*}
			& \mathit{zip}\Big(\Pi\big[\pi \mapsto \play_\calG\big(\dot{s}, \prod_{\agent{i} \in \agents} {\Delta}({\bX(\agent{i})})\big)\big] \Big) \in \calL(\calA).
		\end{align*}
		if and only if $\mathit{even}(\play_{\calP_\Pi}(v_0, \widetilde{\Delta}))$, as required.\qedhere
	\end{itemize}
\end{proof}

\begin{proposition}\label{lem:first_direction}
	$\calP_\Pi$ is won by the existential team if and only if
	\begin{align}\label{eq:statement}
		&\widetilde{\flat}\ldot  \mathit{zip}\Big(\Pi\big[\pi \mapsto \play_\calG\big(\dot{s}, \prod_{\agent{i} \in \agents} f_{\bX(\agent{i})}\big)\big] \Big)\in \calL(\calA).
	\end{align}
\end{proposition}
\begin{proof}
	In $\calP_\Pi$ we quantify over strategies in $\calP_\Pi$ for each player in $[2m]$ (cf.~\Cref{def:cpgSem}).
	Conversely, Eq.~(\ref{eq:statement}) uses the same quantitation order (in the prefix $\widetilde{\flat}$) but quantifies over strategies in $\calG$ via strategy variables $1, \ldots, 2m$ (Eq.~\ref{eq:strat_prefix}).
	Using the translation in \Cref{lem:strat_translation} we can translate any strategy assignment in $\calP_\Pi$ to an equivalent one in $\calG$, and vice versa.
	As the type (i.e., universal or existential) and order of quantification is $\calP_\Pi$ and Eq.~(\ref{eq:statement}) is the same (cf.~Eq.~\ref{eq:strat_prefix}) and we can translate strategies between $\calG$ and $\calP_\Pi$, the result follows. 
\end{proof}

\subsection{Positional Determinacy}

\newcommand{\comStrat}[1]{\mathit{com}(#1)}

The more challenging direction is to show that $\calP_\Pi$ is won by the existential team iff $\calB$ accepts $\mathit{zip}(\Pi)$.
The key challenge is that the way strategies are quantified is fundamentally different: 
In  $\calP_\Pi$ we quantify over full strategies in advance (i.e., ``outside'') and in $\calB$ we (conjunctively or disjunctively) pick actions in each step of the automaton (i.e., ``inside'').
The key ingredient we use is the \emph{positional determinacy} of CPGs. 
Intuitively, a positional strategy is one that decides on an action based solely on the current vertex of the game. 

\begin{definition}
	A \emph{positional strategy} in $\calP = (V, v_0, \moves, m, \mu, c)$ is a function $f : V \to \moves$.
	We write $\posController{\calP}$ for the set of all positional strategies in $\calP$.
\end{definition}

Positional determinacy in the context of (classical) turn-based $2$-player parity games means that players can pick an action based solely on the current vertex of the game. 
In contrast, the quantification over strategies in CPGs can have multiple alternations, the strategy for each player thus depends on the current vertex of the game and the action selected by all strategies quantified before it. 
We will represent the existentially quantified strategies using \emph{Skolem functions} (known, e.g., from first-order logic and sometimes called \emph{dependence map} \cite{MogaveroMPV14}) that get the actions selected by strategies quantified earlier as an explicit input.

\begin{definition}
	A \emph{positional $k$-Skolem strategy} CPG $\calP = (V, v_0, \moves, \allowbreak m, \mu, c)$ is a function $\zeta : V \times \moves^k\to \moves$.
	We write $\skoController{\calP, k}$ for the set of positional $k$-Skolem strategies in $\calP$.
\end{definition}

A positional $k$-Skolem strategy $\zeta$ can pick an action based on the current vertex of the game and $k$ actions that have been selected previously. 
The intuition is that the strategy for player (or strategy variable) $2k$ (which is existentially quantified) can observe the actions selected by the $k$-universally quantified strategies before it.

\begin{definition}
	Assume $\bE$ is a function that maps each $k \in [m]$ to a positional $k$-Skolem strategy in $\calP$ and let $\bO : [m] \to \posController{\calP}$ map each each $k \in [m]$ to a positional strategy in $\calP$. 
	We combine $\bE$ and $\bO$ into a mapping $\comStrat{\bE, \bO} : [2m] \to \strats{\calP}$ as follows.
	For a path $u \in V^+$, we define $\mathit{last}(u) \in V$ as the last vertex in $v$.
	For an odd index $2k-1$ we then define $\comStrat{\bE, \bO}(2k-1) : V^+ \to \moves$ as
	\begin{align*}
		\comStrat{\bE, \bO}(2k - 1)(u) &:= \bO(k)\big(\mathit{last}(u)\big)
	\end{align*}
	For an even index $2k$ we define $\comStrat{\bE, \bO}(2k) : V^+ \to \moves$ as 
	\begin{align*}
		\comStrat{\bE, \bO}(2k)(u) := \bE(k)\big(\mathit{last}(u), (&\bO(1)(\mathit{last}(u)), \\
		&\bO(3)(\mathit{last}(u)), \\
		&\ldots,\\
		 &\bO(2k-1)(\mathit{last}(u)))\big).
	\end{align*}
\end{definition}

The idea of $\comStrat{\bE, \bO}$ is to combine the Skolem strategies in $\bE$ and the positional strategies in $\bO$ into a strategy for each player.
For odd player we simply take the strategy given by $\bO$ (applying it to the last vertex in the given sequence).
For even players, we query the Skolem strategy given by $\bE$ and provide it with the actions that all universally quantified (odd) strategies before it have selected. 
We can now state that CPG are positionally determined by making use of Skolem functions:

\begin{proposition}[\cite{MalvoneMS16}]\label{prop:cpgDet}
	Let $\calP = (V, v_0, \moves, m, \mu, c)$  be a CPG. 
	We have that $\calP$ is \emph{won by the existential players} if and only if  
	\begin{align*}
		&\bigexists_{k \in [m]} \zeta_{k} \in  \skoController{\calP, k}\ldot \bigforall_{k \in [m]} f_{k} \in \posController{\calP}\ldot \\
		&\quad\quad\quad\quad \mathit{even}\Big( \play_\calP(v_0, \prod_{k \in [2m]} \comStrat{\bE, \bO}(k)) \Big)
	\end{align*}
	where $\bE (k):= \zeta_k$ and $\bO(k) :=  f_k$ for $k \in [m]$.
\end{proposition}

\Cref{prop:cpgDet} states that instead of following the quantifier prefix as in \Cref{def:cpgSem} we can instead quantify over Skolem functions $\zeta_1, \ldots, \zeta_{m}$ for all existentially quantified variables. 
Put informally, the proposition thus states that existentially quantified strategies only need to know the current vertex and all actions selected by universally quantified strategies quantified before it. 
A proof \Cref{prop:cpgDet} follows directly from the the construction in \cite[Thm.~4.1]{MalvoneMS16}.

\subsection{The Second Implication}

Using \Cref{prop:cpgDet} we can now prove:

\begin{proposition}\label{lem:second_direction}
	$\calP_\Pi$ is won by the existential team if and only if $\mathit{zip}(\Pi) \in \calL(\calB)$.
\end{proposition}

We prove both directions of \Cref{lem:second_direction} separately (in \Cref{lem:second_direction_first,lem:second_direction_second}).

\begin{lemma}\label{lem:second_direction_first}
	If $\calP_\Pi$ is won by the existential team then $\mathit{zip}(\Pi) \in \calL(\calB)$.
\end{lemma}
\begin{proof}
	Assume that $\calP_\Pi$ is won by the existential player.
	Using \Cref{prop:cpgDet} we thus get positional sklolem functions $\zeta_1,  \ldots, \zeta_{m}$ such that for 
	\begin{align}\label{eq:choiceOfSkolem}
		\begin{split}
			&\bigforall_{k \in [m]} f_{k} \in \posController{\calP}\ldot \mathit{even}\Big( \play_{\calP_\Pi}(v_0,  \prod_{k \in [2m]} \comStrat{\bE, \bO}(k) \Big)
		\end{split}
	\end{align}
	where $\bE(k) := \zeta_k$ and $\bO(k) := f_k$ for $k \in [m]$.
	
	We use the Skolem functions to construct an accepting run DAG $\mathbb{D}$ of $\calB$ on $\mathit{zip}(\Pi)$. 
	We construct $\mathbb{D}$ iteratively.
	We begin with a DAG that consist of the single node $((q_0, \dot{s}), 0)$, i.e., we start with the unique initial state of $\calB$ (note that by definition of run DAGs, nodes are indexed by their current depth, cf.~\Cref{app:prelim}).
	Now assume that there exists some node that we have not visited.
	Let this note be $((q, s), N)$. 
	We observe that $(q, s, N)$ is a vertex in $\calP_{\Pi}$.

	The $N$th input read by $\calB$ on $\mathit{zip}(\Pi)$ is $\mathit{zip}(\Pi)(N)$.
	By definition of $\calB$ the transition function from $(q, s)$ on the $N$th input is
	\begin{align}\label{eq:prefix}
		&\delta'\big((q, s), \mathit{zip}(\Pi)(N) \big) = \bigwedge_{a_1 \in \moves} \bigvee_{a_2 \in \moves}  \cdots \bigvee_{a_{2m} \in \moves} (q', s')
	\end{align}
	where $q' = \delta(q, \mathit{zip}(\Pi)(N)[\pi \mapsto s])$ and $s' = \kappa\big(s, \prod_{\agent{i} \in \agents}  a_{\bX(\agent{i})}  \big)$.
	
	We need to add children of the node $((q, s), N)$ that full the above formula.
	We will construct a set $Y \subseteq Q \times S$ that is a model for the above, i.e., $Y \models \delta'\big((q, s), \mathit{zip}(\Pi)(N) \big)$. 
	We will construct $Y$ by following all the action selection in the prefix of $\delta'\big((q, s), \mathit{zip}(\Pi)(N) \big)$ and construct an intermediate set of functions $Z \subseteq ([2m] \rightharpoonup \moves)$.
	We can think of $Z$ as combinations of actions $a_1, \ldots, a_{2m}$ that we select in the prefix of Eq.~(\ref{eq:prefix}).
	This set $Z$ is constructed in accordance with the Skolem functions $\zeta_1, \ldots, \zeta_m$, i.e., for every disjunctive choice we will selected the action that is picked by the respective Skolem function. 
	Initially, we set $Z = \{\{\}\}$ as the singleton set containing only the empty function $\{\} : [2m] \rightharpoonup \moves$. 
	For every $k$ from $1$ to $2m$ we do the following: 
	If $k$ is odd, so action $a_k$ is chosen conjunctively we add all possible actions. 
	That is, we update $Z := \{ \bA[k \mapsto a] \mid \bA \in Z, a \in \moves \}$.
	If $k$ is even, so $k = 2k'$ for some $k'$, we can pick one action for $k$ for each element in $Z$.  
	Given $\bA \in Z$ (so $\bA$ is a function $\{1, \ldots, k-1\} \to \moves$) we define the action $a_\bA$ as
	\begin{align*}
		a_\bA := \zeta_{k'}\big((q, s, N), (\bA(1), \bA(3), \ldots, \bA(2k'-1) )\big).
	\end{align*}
	That is, we use the $k'$th Skolem functions and query it with the current vertex $(q, s, N)$ and the action selected  by all previously chosen conjunctive action choices (at odd positions) in $\bA$.
	We then update $Z := \{ \bA[k \mapsto a_\bA ] \mid \bA \in Z\}$.
	
	After repeating this procedure, we have constructed a set $Z \subseteq ([2m] \to \moves)$ that maps all players to actions. 
	Informally, this corresponds to a possible assignment of the actions selected in the prefix of $\delta'\big((q, s), \mathit{zip}(\Pi)(N) \big)$.
	For each $\bA \in Z$ we can define a unique $s_\bA$ by following the construction in the definition of $\calB$.
	That is, we define $s_\bA = \kappa\big(s, \prod_{\agent{i} \in \agents}  \bA({\bX(\agent{i})})  \big)$.
	We then define 
	\begin{align*}
		Y := \Big\{ \big( \delta(q, \mathit{zip}(\Pi)(N)[\pi \mapsto s]), s_\bA\big) \mid \bA\in Z \Big\}.
	\end{align*}
	It is easy to see that $Y \models \delta'\big((q, s), \mathit{zip}(\Pi)(N) \big)$: In the selection of actions (when constructing $Z$) we have consider all possible actions for each conjunctive choice and picked a particular action for each disjunctive choice. 
	
	In our run DAG we now add a node $((q', s'),N+1)$ for each $(q', s') \in Y$, and add an edge from $((q, s), N)$ to $((q', s'),N+1)$.
	
	Let $\mathbb{D}$ be the infinite DAG obtained by following this construction. 
	It is easy to see that $\mathbb{D}$ is a valid run of $\calB$ on $\mathit{zip}(\Pi)$ (as we have argued above all children that we add to any given node satisfy the transition formula).  
	
	It remains to argue that $\mathbb{D}$ is accepting. 
	For this we consider an arbitrarily infinite path $\tau$ in $\mathbb{D}$, s.t., 
	\begin{align*}
		\tau = ((q_0, s_0), 0) ((q_1, s_1), 1) ((q_2, s_2), 2) \cdots
	\end{align*}
	We define an analogous path in $\calP_\Pi$ as follows (by simply regrouping parenthesis):
	\begin{align*}
		\tau' = (q_0, s_0, 0) (q_1, s_1, 1) (q_2, s_2, 2) \cdots
	\end{align*}
	In construction $\mathbb{D}$, we always mimicked the choice for each distinctively chosen action by using the Skolem functions $\zeta_1, \ldots, \zeta_{m}$.
	We thus get that $\tau'$ is a play in $\calP_{\Pi}$ under these Skolem functions:
	That is, there exist $f_1, f_2, \ldots, f_m \in \posController{\calP_\Pi}$ such that 
	\begin{align*}
		\play_{\calP_\Pi}(v_0, \prod_{k \in [2m]} \comStrat{\bE,\bO}(k)) = \tau'
	\end{align*}
	where $\bE(k) := \zeta_k$ and $\bO(k) :=  f_k$ for $k \in [m]$.
	By the choice of $\zeta_1, \ldots, \zeta_{m}$ (cf.~Eq.~(\ref{eq:choiceOfSkolem})) we thus get that $\mathit{even}(\tau')$.
	
	Now the sequence of colors traversed in $\tau'$ is the same as the sequence of colors in the path $\tau$ in the run DAG. 
	The minimal color that appears infinitely often in $\tau$ is thus also even and so the path is accepting. 
	As this holds for all paths, $\mathbb{D}$ is accepting. 
	
	We have constructed an accepting run DAG of $\calB$ on $\mathit{zip}(\Pi)$, so $\mathit{zip}(\Pi) \in \calL(\calB)$ as required. 
\end{proof}

\begin{lemma}\label{lem:second_direction_second}
	If $\mathit{zip}(\Pi) \in \calL(\calB)$ then $\calP_{\Pi}$ is won by the existential team.
\end{lemma}
\begin{proof}
	Assume that $\mathit{zip}(\Pi) \in \calL(\calB)$, and let $\mathbb{D}$ be an accepting run DAG of $\calB$.
	We will use this DAG to construct Skolem functions $\zeta_1, \ldots, \zeta_{m}$ such that 
	\begin{align}\label{eq:SkolemTarget}
		\begin{split}
			&\bigforall_{k \in [m]} \!\! f_{k} \in \posController{\calP}\ldot \mathit{even}\Big( \play_{\calP_\Pi}(v_0, \prod_{k \in [2m]} \comStrat{\bE, \bO}(k)) \Big)
		\end{split}
	\end{align}
	where $\bE(k) :=\zeta_k$ and $\bO(k) := f_k$ for $k \in [m]$.
	By \Cref{prop:cpgDet} this would imply that $\calP_{\Pi}$ is won by the existential team.

	For any node $x = ((q, s),N)$ in the run DAG $\mathbb{D}$, the children of $((q, s),N)$ satisfy 
	\begin{align}\label{eq:prefix2}
		&\delta'\big((q, s), \mathit{zip}(\Pi)(N) \big) = \bigwedge_{a_1 \in \moves} \bigvee_{a_2 \in \moves}  \cdots \bigvee_{a_{2m} \in \moves} (q', s')
	\end{align}
	where $q' = \delta(q, \mathit{zip}(\Pi)(N)[\pi \mapsto s])$ and $s' = \kappa\big(s, \prod_{\agent{i} \in \agents}  a_{\bX(\agent{i})}  \big)$.
	
	We will construct functions $\alpha^{x}_1, \ldots, \alpha^{x}_m$ where $\alpha^{x}_k : \moves^k \to \moves$ that will pick a choice for each disjunction based on the previous choices for each conjunction. 
	We first define $\alpha^{x}_1 : \moves \to \moves$:
	For any $a^1 \in \moves$ we define $\alpha^{x}_1(a^1)$ by \emph{fixing} the first conjunctively chosen action $a_1$ in Eq.~(\ref{eq:prefix2}) to be $a^1$.
	As $a_1$ was chosen conjunctivley, the children of $x = ((q, s),N)$ still satisfy Eq.~(\ref{eq:prefix2}) with $a_1$ fixed to $a^1$. 
	Now define $a^2 \in \moves$ as some valid choice for the (disjunctively chosen) $a_2$. 
	That is, we define $a^2$ such that the children of $x = ((q, s),N)$ in $\mathbb{D}$ still form a model of Eq.~(\ref{eq:prefix2}) with actions $a_1 := a^1, a_2 := a^2$ fixed.
	We set $\alpha^{x}_1(a^1) := a^2$.
	After having defined $\alpha^{x}_1 $ we can define $\alpha^{x}_2(a^1, a^3)$: We fix the first conjunctively chosen action be $a^1$, the second disjunctive action be $\alpha^{x}_1(a^1)$ and the second conjunctive action be $a^3$, and the define $\alpha^{x}_2(a^1, a^3)$ as some action that we can select for the second disjunctive choice such that the children of $x = ((q, s),N)$ in $\mathbb{D}$ satisfy the subformula with those actions fixed. 
	The construction of the remaining $\alpha^{x}_1, \ldots, \alpha^{x}_m$ is analogous. 
	Intuitively, each $\alpha^{x}_k$ serves as a Skolem function for the $k$th disjunctive action by fixing an action based solely on the earlier conjunctive actions. 
	Together they guarantee that by following the selected action by each Skolem function for each disjunctive choice, we always reach some child of $x = ((q, s),N)$ in $\mathbb{D}$.
	
	We can now define the desired positional Skolem functions $\zeta_1, \ldots, \zeta_{m}$ for $\calP_\Pi$.
	We define $\zeta_k$ as follows: Given any vertex $v = (q, s, N)$ in $\calP_\Pi$ and actions $a_1, \ldots, a_k \in \moves$ (corresponding to the actions selected by all universally quantified strategies before $2k$) we check if $x := ((q, s), N)$ is a node in $\mathbb{D}$.
	If there does not exists such a node we can return an arbitrary action (we will later argue that this situation will never be reached in any play $\calP_{\Pi}$).
	Otherwise $x$ is a node in $\mathbb{D}$ and we get the Skolem functions $\alpha^{x}_1, \ldots, \alpha^{x}_m$ we have constructed earlier. 
	We define $\zeta_k(v, (a^1, \ldots, a^k)) := \alpha^{x}_k(a^1, \ldots, a^k)$, i.e., we select the action that $\alpha_k^{x}$ picks for the $k$th disjunction when using the actions $a^1, \ldots, a^k$ for the previous $k$ conjunctions. 
	
	It remains to argue that the Skolem functions $\zeta_1, \ldots, \zeta_{m}$ fulfill Eq.~(\ref{eq:SkolemTarget}).
	Let $f_1, \ldots. f_{m} \in \posController{\calP_{\Pi}}$ and consider the resulting play 
	\begin{align*}
		\tau &= \play_{\calP_\Pi}\big(v_0, \prod_{k \in [2m]}  \comStrat{\bE, \bO}(k)\big) \\
		&= (q_0, s_0, 0)  (q_1, s_1, 1) (q_2, s_2, 2)  \cdots 
	\end{align*}
	where $\bE(k) := \zeta_k$ and $\bO (k):= f_k$ for $k \in [m]$.
	We consider the equivalent path in the run DAG (by regrouping parenthesis):
	\begin{align*}
		\tau' = ((q_0, s_0), 0)  ((q_1, s_1), 1)  ((q_2, s_2), 2)  \cdots .
	\end{align*} 
	By construction of $\zeta_1, \ldots, \zeta_m$ we always pick the actions that are disjunctively chosen in accordance with $\mathbb{D}$. 
	It is therfore easy to see that $\tau'$ is an infinite path in $\mathbb{D}$.
	By the assumption that  $\mathbb{D}$ is accepting the minimal color that appears infinitely often is even. 
	As the automaton state sequence agrees with that of $\tau$, we thus get that $\mathit{even}(\tau)$ holds. 
	So $\zeta_1, \ldots, \zeta_m$ satisfy Eq.~(\ref{eq:SkolemTarget}). 
	By  \Cref{prop:cpgDet} this implies that $\calP_{\Pi}$ is won by the existential team, as required. 
\end{proof}

By \Cref{lem:first_direction}, we get that $\Pi$ satisfies Eq.~(\ref{eq:lemmaTarget}) iff $\calP_\Pi$ is won by the existential team (\Cref{lem:first_direction}).
By \Cref{lem:second_direction}, $\calP_\Pi$ is won by the existential team iff $\mathit{zip}(\Pi) \in \calL(\calB)$.
Consequently, $\Pi$ satisfies Eq.~(\ref{eq:lemmaTarget}) iff $\mathit{zip}(\Pi) \in \calL(\calB)$, proving \Cref{lem:taregt_for_B} and thus of \Cref{prop:sim}.

\section{Details on the Section \ref{sec:imp}}\label{app:exp}

In this section, we provide additional details on the formulas checked in \Cref{sec:imp}.

\subsection{Details on Section \ref{sec:sub:atls}}

We check the following \HyperSLP{} formula: 
\begin{align*}
	&\exists x\ldot \forall x_1, \ldots, x_n\ldot \left(\bigwedge_{i=1}^n \ltlG \big( \langle \mathit{wt}, i\rangle_\pi \to \ltlF \neg \langle \mathit{wt}, i\rangle_\pi \big)\right)\\
	&\quad\quad\quad\quad\quad\quad\quad\quad\left[\pi: (\agent{sched} \mapsto x, \agent{y_1} \mapsto x_1, \ldots,  \agent{y_n} \mapsto x_n)	\right]
\end{align*}
This formula states that the scheduling agent $\agent{\mathit{sched}}$ has a strategy such that none of the working agents  $\agent{y_1}, \ldots, \agent{y_n}$ starves, i.e., whenever agent $\agent{y_i}$ waits for a grant (modeled by proposition $\langle \mathit{wt}, i\rangle$), it will eventually not wait any more.
Note that this formula is equivalent to the SL[1G] specification used by \citet{CermakLM15}.
We check it for various values of $n \in \nat$ and give the results in \Cref{tab:results-sl}.

\subsection{Details on Section \ref{sec:sub:hyper}}

In \Cref{sec:sub:hyper}, we check random formulas from a range of different families. 
We assume we are given a CGS $\calG$ over atomic proposition $\ap$ and agents $\agents$ (These CGS are obtained automatically from the ISPL models from \cite{LomuscioQR09}).

\paragraph{Security (Sec).}

We select some agent $\agent{i} \in \agents$ and some AP $g \in \ap$ modeling a goal, an AP $h \in \ap$ modelling a high-security input, and an AP $o \in \ap$ modeling an observable output. 
We then construct the following \HyperSLP{} formula:
\begin{align*}
	&\exists x\ldot \forall y_1, \ldots, y_{i-1}, y_{i+1}, \ldots, y_n\ldot \exists z_1, \ldots, z_n\ldot \\
	&\quad\quad(\ltlF g_{\pi} \land \ltlG (o_{\pi} \leftrightarrow o_{\pi'}) \land \ltlF (h_{\pi} \not\leftrightarrow h_{\pi'}))\\
	&\quad\quad\quad\quad\left[\begin{aligned}
		&\pi: (y_1, \ldots, y_{i-1}, x, y_{i+1}, \ldots, y_n)\\
		&\pi':(z_1, \ldots, z_n)
	\end{aligned}	\right]
\end{align*} 
This formula states that $\agent{i}$ has a strategy to eventually reach $g$.
Moreover, it may not leak all information about $h$ via $o$. 
We model this using the idea of non-inference \cite{McLean94}.
The idea is that the behavior in $o$ should not leak $h$, so there must be ``plausible deniability''. 
That is, the same observation via $o$ is also possible for some different input sequence via $h$. 

Concretely, $\agent{i}$ should be able to reach $g$ on path $\pi$ no matter what the other agents play.
In addition, there must exists some path $\pi'$ (which we state by quantifying the strategies of all agents existentially), that has the same observations $\ltlG (o_{\pi} \leftrightarrow o_{\pi'})$ but a different high-security input ($ \ltlF (h_{\pi} \not\leftrightarrow h_{\pi'})$).

\paragraph{Good-Enough Synthesis (GE)}

In many situations, asking for a strategy that wins in all situations is too restrictive. 
Instead, it often suffices to look for strategies that are \emph{good-enough} (GE), i.e., strategies that win on every possible input sequence for which there exists a winning output sequence \cite{AlmagorK20,AminofGR21}. 
We can express this formally using \HyperSLP{}. 
Concretely, assume that $\varphi = \quant_1 x_1, \ldots, \quant_m x_m\ldot \psi[\pi : \bX]$ is any \HyperSL{} formula over a single path variable (or, equivalently, a \SL[1G] formula). 
We want to express that that $\varphi$ only needs to holds on traces with input $i \in \ap$, on which some path with the same input actually satisfies $\psi$. 

We can express that $\varphi$ is a GE-strategy as follows:
\begin{align*}
	&\quant_1 x_1, \ldots, \quant_m x_m\ldot \forall y_1, \ldots, y_m\ldot  \\
	&\quad\quad\quad(\ltlG (i_{\pi} \leftrightarrow i_{\pi'}) \land \psi[\pi'/\pi]) \to \psi
\end{align*} 
where we write $\psi[\pi'/\pi]$ for the formula with all occurrences of $\pi$ replaced with $\pi'$.

In the formula, we quantify over strategies $x_1, \ldots, x_m$ as in $\varphi$ and use these strategies to construct path $\pi$. 
Afterwards, we universally quantify over any path $\pi'$ in the system by picking strategies $y_1, \ldots, y_n$ for all agents. 
We then state that $\psi$ only needs to hold on $\pi$ provided $\pi'$ has the same input and satisfies $\psi$.
Phrased differently, $\pi$ only needs to win, provided some path with the same inputs can ensure $\psi$. 
Note that, depending on the prefix in $\varphi$, this is not expressible in weaker hyperlogics such as \HyperATLS{} and \HyperATLSS{}.

\paragraph{Random (Rnd).}

For the random category, we use a random LTL formula (sampled using \texttt{spot} \cite{Duret-LutzRCRAS22}) and add a prefix of quantifiers to yield a \HyperSLP{} formula.

\fi

\end{document}